\documentclass[11pt]{amsart}

\usepackage[top=1in, left=1in, right=1in, bottom=1in]{geometry}

\usepackage{amsmath,amssymb,amsthm,amsfonts}
\usepackage{paralist}
\usepackage{booktabs}
\usepackage[usenames,dvipsnames]{color}
\usepackage{comment}
\usepackage{graphicx,epsf}
\usepackage{graphics}
\usepackage{epsfig}
\usepackage{epstopdf}
\usepackage{psfrag}

\usepackage{amsmath,amsfonts,amssymb,latexsym,amsthm}
\usepackage{graphicx,epsf}
\usepackage{amsthm}
\usepackage{multicol}
\usepackage{ytableau}

\usepackage{url}
\usepackage{hyperref} 

\usepackage{tikz}

\usepackage[colorinlistoftodos]{todonotes}

\RequirePackage{cleveref}
\usepackage{hypcap}
\hypersetup{colorlinks=true, citecolor=darkblue, linkcolor=darkblue}
\definecolor{darkblue}{rgb}{0.0,0,0.7}
\newcommand{\darkblue}{\color{darkblue}}

\definecolor{darkred}{rgb}{0.68,0,0}

\definecolor{darkgreen}{rgb}{0,.38,0}

\newcommand{\defn}[1]{\emph{\darkblue #1}}

\newlength{\ml}
\newtheorem{thm}{Theorem}[section]

\newtheorem{cor}[thm]{Corollary}
\newtheorem{prop}[thm]{Proposition}

\newtheorem{conj}[thm]{Conjecture}

\newtheorem{ex}[thm]{Example}

\theoremstyle{plain}

\theoremstyle{definition}

\newtheorem{ques}{Question}
\newtheorem{rem}[thm]{Remark}

\numberwithin{equation}{section}

\def\bx{{\textbf{\textit{x}}}}
\def\by{{\textbf{\textit{y}}}}

\newcommand{\la}{\lambda}

\newcommand{\cP}{\mathcal P}

\newcommand{\nn}{\Bbb N}

\newcommand{\ga}{\gamma}

\newcommand{\al}{\alpha}
\newcommand{\be}{\beta}

\def\sgn{\mathrm{sgn}}

\def\.{\hskip.06cm}

\def\SSYT{ {\text {\rm SSYT}  } }

\def\P{{{\rm{\textsf{P}} }}}

\def\FP{{\rm{\textsf{FP}}}}

\def\SP{{{\rm{\textsf{\#P}}}}}

\def\GapP{{{\rm{\textsf{GapP}} }}}

\def\FP{{{}\rm{\textsf{FP}} }}
\def\NP{{{\rm{\textsf{NP}} }}}

\def\SBQP{{\rm{\textsf{\#BQP}}}}
\def\QMA{{\rm{\textsf{QMA}}}}
\newcommand{\poly}{poly}

\def\bbb{{\text{\bf b}}}

\def\bbj{{\text{\bf j}}}

\def\<{\langle}
\def\>{\rangle}

\def\det{\mathrm{det}}

\def\sgn{\mathrm{sgn}}

\def\la{\lambda}

\begin{document}


\title[Polynomial time algorithms for multiplicities]{Polynomial time classical versus quantum algorithms for representation theoretic multiplicities}

\author{Greta Panova}
\address{Department of Mathematics, University of Southern California, 
Los Angeles, CA 90089}
\email{gpanova@usc.edu}
\urladdr{https://sites.google.com/usc.edu/gpanova/}
\thanks{The author was partially supported  by the NSF and the Simons foundation.}

\date{\today}

\begin{abstract}
Littlewood--Richardson, Kronecker and plethysm coefficients are fundamental multiplicities of interest in Representation Theory and Algebraic Combinatorics. Determining a combinatorial interpretation for the Kronecker and plethysm coefficients is a major open problem, and prompts the consideration of their computational complexity. Recently it was shown that they behave relatively well with respect to quantum computation, and for some large families there are polynomial time quantum algorithms \cite{LH24} (also~\cite{bravyi2024quantum}). In this paper we show that for many of those cases the Kronecker and plethysm coefficients can also be computed in polynomial time via classical algorithms, thereby refuting some of the conjectures in \cite{LH24}. This vastly limits the  cases in which the desired super-polynomial quantum speedup could be achieved. 
\end{abstract}

\keywords{Kronecker coefficients, plethysm coefficients, Symmetric group representations, Computational Complexity, polynomial time algorithms}

\maketitle

\section{Introduction}

Some of the outstanding open problems in Algebraic Combinatorics concern finding ``combinatorial interpretations'' for certain representation-theoretic multiplicities and other structure constants which are naturally nonnegative integers. While ``combinatorial interpretation" is a loosely defined term generally assumed to mean ``counting some nice objects'', a more formal definition would go through computational complexity theory with the premise that such a nice positive formula normally implies that these counting problems are in $\SP$. In particular, showing no nice combinatorial interpretation exists could be done by showing that the problem is not in $\SP$ under standard computational complexity assumptions, see~\cite{Pak22,panova2023computational}. Note that all quantities in question are already  in the, conjecturally strictly larger, class $\GapP_{\geq 0}:=\{ f -g: f,g \in \SP, f-g \geq 0\}$, i.e. nonnegative functions which can be written as differences of two $\SP$ functions, and deciding their positivity is $\NP$-hard~\cite{FI20,IMW17}. In contrast with classical computation, some of these multiplicities are shown to belong to the $\SBQP$ class, the quantum analogue of $\SP$, and deciding positivity is in $\QMA$, that is, there exists a polynomial time quantum verifier for their positivity, see~\cite{bravyi2024quantum,CHW,ikenmeyer2023remark}. 

In~\cite{LH24}, following~\cite{bravyi2024quantum}, the authors exhibited efficient quantum algorithms for computing these multiplicities in certain cases (based on dimensions), and conjectured that there would not be such efficient classical algorithms. Here we disprove some of these conjectures. We show that for a large family of parameters the multiplicities can actually be computed in polynomial time. This shows that \emph{the desired super-polynomial quantum speedup cannot be achieved} for those families. We then pose further conjectures on the existence of algorithms of particular runtimes. Our main intuition arises from the asymptotic behaviors of dimensions and multiplicities in the various regimes and characterization of the families of partitions.

To be specific, let $V_\la$ be the Weyl modules arising from the irreducible polynomial representations $\rho_\la$ of the $GL_N(\mathbb{C})$ for integer partitions $\la$ with at most $N$ nonzero parts. Let $\mathbb{S}_\la$ be the Specht modules, i.e. the irreducible representations of the symmetric group $S_n$ indexed by partitions $\la \vdash n$, and denote by $f^\la$ the dimension of $\mathbb{S}_\la$. The Littlewood--Richardson coefficients $c^\la_{\mu\nu}$ are defined as the multiplicities of $V_\la$ in the tensor product $V_\mu \otimes V_\nu$, that is
\begin{align}
V_\mu \otimes V_\nu = \bigoplus_{\la \vdash |\mu|+|\nu|} V_\la^{\oplus c^\la_{\mu\nu}}.
\end{align}
Let $g(\la,\mu,\nu)$ be the Kronecker coefficient given as the multiplicity of $\mathbb{S}_\la$ in $\mathbb{S}_\mu \otimes \mathbb{S}_\nu$, where $S_n$ acts diagonally, so
\begin{align}
\mathbb{S}_\mu \otimes \mathbb{S}_\nu = \bigoplus_{\la \vdash n} \mathbb{S}_\la^{\oplus g(\la,\mu,\nu)}.
\end{align}
The plethysm coefficients $a^\la_{\mu \nu}$ are defined as the multiplicities of $V_\la$  in the composition $\rho_\mu(\rho_\nu)$. We also consider the Kostka numbers $K_{\la\mu}$, which are multiplicities of the weight $\mu$ space in $V_\la$. All of these coefficients can be defined purely combinatorially using symmetric functions and tableaux, see Section~\ref{sec:defn}.

While Kostka and Littlewood--Richardson coefficients are known to count certain tableaux, finding combinatorial interpretations for plethysm and Kronecker coefficients are major open problems, see~\cite{colmenarejo2022mystery,Pak22,Pan23,Sta00}. A combinatorial interpretation usually implies that verifying positivity is ``easy'', that is, if we exhibit one object among the ones they are counting, there is a polynomial time algorithm which checks that this is the right object. Computing them does not have to be efficient, and in particular it would be at least exponential in general as they are $\SP$-hard (assuming the \emph{exponential time hypothesis} of \cite{impagliazzo2001complexity}). However, in many cases there are efficient (polynomial time algorithms) beyond the ones described in~\cite{CDW,PPcomp}. Here we show that

\begin{thm}\label{thm:kron1}
Let $\la,\mu,\nu \vdash n$ and suppose that $f^\nu \leq n^k$ for some $k$. Then $g(\la,\mu,\nu)$ can be computed in time $O(D(k) n^{4k^2+1}\log(n))$, where $D(k) = (4k)^{8(8k^4+k^2)}$. In particular, if $\la^{(n)}, \mu^{(n)},\nu^{(n)}$ are families of partitions of $n$, such that $f^{\nu^{(n)}} \leq n^k$ for a fixed constant $k$ then $g(\la^{(n)},\mu^{(n)},\nu^{(n)})$ can be computed in polynomial time $O(n^{4k^2+1}\log(n))$. 
\end{thm} 
In particular, this refutes Conjecture 2 of~\cite{LH24} and partially answers the discussion in~\cite{bravyi2024quantum} (after Lemma 2). There are  polynomial time algorithms for computing Kronecker coefficients when all three partitions have constant lengths, see~\cite{CDW,PPcomp}, however here we have no restrictions on two of the partitions. 

\begin{thm}\label{thm:pleth}
Let $d$, $m$ be integers, $n=dm$ and $\la \vdash n$, such that $\la_1 \geq \ell(\la)$. Then the plethysm coefficient $a^\la_{d,m}$ can be computed in time
\begin{enumerate}
\item $O(n^{d\ell}\log(d))$ where $\ell=\ell(\la)$. 
\item $O(n^{4K^3(K+1)}\log(d))$ where  $f^\la \leq n^k$ and $K=4k^2$ for arbitrary $d,m$.
\end{enumerate}
In particular, we have a polynomial time algorithm for computing $a^\la_{d,m}$ if either $d$ and $\ell(\la)$ are fixed, or $d$ grows but the dimension $f^\la$ grows at most polynomially.
\end{thm}

In particular, the second case refutes Conjecture 1~\cite{LH24} for the case when $\mu=(d)$ and $\nu=(m)$, as then both classical and quantum algorithms run in polynomial time $O(f^\la)$. Polynomial time algorithms when $d$ is fixed are also given in~\cite{kahle2016plethysm}. 

The main results of \cite{LH24} give  quantum algorithms for computing $g(\la,\mu,\nu)$ in time $O\left( \frac{f^\nu f^\mu}{f^\la} \right)$ (also stated in~\cite{bravyi2024quantum}), plethysm coefficient $a^\la_{\mu,\nu}$ in time $O\left(\frac{f^\la}{(f^\nu)^{|\mu|} f^\mu}\right)$, Kostka numbers $K_{\la,\mu}$ in times $O(f^\la)$ and Littlewood--Richardson coefficient $c^\la_{\mu\nu}$ in time $O\left(\frac{ f^\la}{f^\mu f^\nu} \right)$. The authors show that the Kostka numbers can also be computed by a classical algorithm with the same efficiency, and conjecture that the Littlewood--Richardson coefficients can also be computed by a classical algorithm with runtime $O\left(\frac{ f^\la}{f^\mu f^\nu} \right)$, but conjectured that the analogy would not hold for Kronecker and plethysm coefficients. Here we generalize the results for Kostka numbers, and explore the computation of Littlewood--Richardson coefficients. As we disprove some of the~\cite{LH24} conjectures about the Kronecker and plethysm coefficients, we pose the opposite conjecture, which is true in many cases, most of them for trivial reasons, see Section~\ref{sec:rem_kron}.

\begin{conj}\label{conj:kron_dim}
Let $\la,\mu,\nu \vdash n$ and suppose that $f^\la \geq f^\mu \geq f^\nu$. The Kronecker coefficient can be computed by a classical algorithm in time $O(\frac{ f^\mu f^\nu}{f^\la} \poly(n) )$. 
\end{conj}

We suspect the plethysm coefficients could also be computed  in time $O\left( \frac{f^\la}{(f^\nu)^{|\mu|} f^\mu} \right)$, see Section~\ref{sec:rem_pleth} for a discussion, and we pose it as a question. 

\begin{ques}\label{q:pleth}
Let $\la \vdash n$, $\mu \vdash d$, $ \nu \vdash m$, such that $dm=n$. Does there exist a classical algorithm computing $a^\la_{\mu \nu}$  running in time $O\left( \frac{f^\la}{(f^{\nu})^d f^\mu}  \poly(n) \right)$?
\end{ques}

Our analysis starts with characterizing the partitions $\nu\vdash n$ for which $f^\nu$ is of polynomial size in Propositions\ref{prop:fixed_aft2} and~\ref{prop:fixed_aft}. We are able to describe such partitions as the ones for which ${\rm aft}(\la):=|\la|-\la_1$ (assuming by symmetry that $\la_1\geq \ell(\la)$) is fixed. However, it is not clear how to characterize all regimes considered in~\cite{LH24}, as the dimensions can have polynomial, exponential and superexponential growths, but in the considered ratios the leading terms could cancel.

\begin{ques}\label{q:triples}
Characterize the triples of partitions $(\la,\mu,\nu)$ of $n$, such that if $f^\la \geq f^\mu \geq f^\nu$ then $1\leq \frac{f^\mu f^\nu}{f^\la} \leq n^k$ for some fixed integer $k$.  
\end{ques}
The first condition is necessary in order to have $g(\la,\mu,\nu)>0$. Similar questions pertain to regimes for polynomially large nonzero Littlewood--Richardson coefficients as discussed in Section~\ref{sec:dimension}. 

\subsection*{Paper outline} 
In Section~\ref{sec:defn} we describe all the necessary concepts and definitions from Algebraic Combinatorics and Computational Complexity as well as simple asymptotic tools. In Section~\ref{sec:dimension} we characterize the partitions for the various dimension regimes of interest in this paper. In Section~\ref{sec:kostka} 
we extend some of the results in~\cite{LH24} to the computation of skew Kostka numbers, and certain cases of Littelwood-Richardson coefficients, and extend further questions on the efficiency of their computation. In Section~\ref{sec:kron} we prove Theorem~\ref{thm:kron1} and in Section~\ref{sec:pleth} we prove Theorem~\ref{thm:pleth}. We conclude with remarks about previous results, further open problems, and discussions on combinatorial and complexity-theoretic implications.

\subsection*{Acknowledgements}
The author is grateful to Vojtech Havlicek and Martin Larocca for many fruitful discussions on the current topic. We also thank Christian Ikenmeyer, Allen Knutson and Igor Pak for insightful conversations on these topics. We thank Anne Schilling, Sergey Fomin and Maxim van den Berg for providing interesting references, and the anonymous referees for many valuable comments. The majority of this work was completed while the author was a visiting member at the Institute for Advanced Study, Princeton, partially supported by the NSF and a Simons Fellowship.

\section{Background and definitions}\label{sec:defn}

We recall some basic definitions and formulas from the theory of symmetric functions and representations of $GL_N$ and $S_n$. For details on the combinatorial sides see~\cite{Mac,S1} and for the representation theoretic aspects see~\cite{Fu97,Sag}.

\subsection{Young tableaux} \label{ss:basics}
Let $\la=(\la_1,\la_2,\ldots,\la_\ell)$ be a \defn{partition}
of size $n:=|\la|=\la_1+\la_2+\ldots+\la_\ell$, where
$\la_1\ge \la_2 \ge \ldots \ge\la_\ell\ge 1$.  We write
$\la\vdash n$ to indicate it is a partition of $n$, and $\cP=\{\la\}$ for
the set of all partitions.  The length of~$\la$, $\ell(\la):=\ell$, is its number of nonzero parts.  Let $p(n) =\#\{\la \vdash n\}$ be the number of partitions of $n$. The famous Hardy-Ramanujan asymptotics gives
\begin{align}\label{eq:pn}
p(n) \sim \frac{1}{4n\sqrt{n}} \exp\left(\pi \sqrt{\frac{2n}{3}} \right)
\end{align}


\ytableausetup{boxsize=2ex}

A \defn{Young diagram} of
\emph{shape}~$\la$ is an arrangement of squares
$(i,j)\in \nn^2$ with $1\le i\le \ell(\la)$
and $1\le j\le \la_i$. The \defn{conjugate} partition $\la'$ is the partition whose Young diagram is the diagonally transposed diagram of $\la$. We denote by $\la=(a_1,a_2,\ldots \mid b_1,b_2,\ldots)$ the partition in \defn{Frobenius coordinates}, that is $a_i = \la_i-i$ and $b_i = \la'_i-i$ are the \defn{arm} and \defn{leg lengths} of the boxes on the diagonal. We denote by $(a^b) = (\underbrace{a,\ldots,a}_b)$ the rectangular Young diagram with $b$ rows of length $a$. 
A \defn{semistandard Young tableau} (SSYT) $T$ of
\emph{shape}~$\la$ and \emph{weight} (or \emph{type})~$\al$ is an
assignment of integers to the squares of the Young diagram of~$\lambda$, so that there are $\al_k$ many integers equal to~$k$, and which weakly increase along rows and strictly
increase down columns, i.e. $T(i,j) \leq T(i,j+1)$ and $T(i,j) \leq T(i+1,j)$. For example, 
\ytableaushort{11244,2235,45}\  is an SSYT of shape $\la=(5,4,2)$ and type $\al=(2,3,1,3,2)$.
Denote by $\SSYT(\la,\al)$ the
set of such tableaux. The \defn{Kostka number} is defined as
$K_{\la,\al} := \bigl|\SSYT(\la,\al)\bigr|$. A \defn{standard Young tableau} (SYT) of shape $\la\vdash n$ is an SSYT of type $(1^n)$, and we have $f^\la := K_{\la,1^n}$, which can be computed by the hook-length formula:
\begin{align}\label{eq:HLF}
f^\la = \frac{ n!}{\prod_{u \in \la } h_u}, \tag{HLF}
\end{align} 
where $u=(i,j)$ goes over all boxes of $\la$ and $h_u$ is the hook length of $u$, that is $h_u=\la_i-i+\la'_j-j+1$. A skew shape $\la/\mu$ is the young diagram obtained by removing the diagram of $\mu$ from $\la$, e.g. $(4,3,3)/(2,1)= \ydiagram{2+2,1+2,3}$. We define SYT and SSYT of shape $\la/\mu$ the same way, and set $f^{\la/\mu}$ to be number of skew SYTs of shape $\la/\mu$ and $K_{\la/\mu,\nu}$ the number of skew SSYTs of shape $\la/\mu$ and type $\nu$.

The irreducible representations of the \defn{symmetric group} $S_n$ are the \defn{Specht modules} $\mathbb{S}_\la$ and are indexed by partitions $\la \vdash n$. A basis for $\mathbb{S}_\la$ can be indexed by the SYTs. In particular
$$\dim \mathbb{S}_\la = f^\la.$$

%
  

The irreducible polynomial representations of  $GL_N(\mathbb{C})$ are the \defn{Weyl modules} $V_\la$ and are indexed by all partitions $\lambda$ with $\ell(\la) \leq N$. The dimension of the weight $\mu$ subspace of $V_\la$ is equal to $K_{\la,\mu}$. 

\subsection{Symmetric functions}\label{ss:sym_func}
Let $\Lambda[\bx]$ be the ring of \defn{symmetric functions} $f(x_1,x_2,\ldots)$ over $\mathbb{Q}$, where $\bx=(x_1,x_2,\ldots)$,
 the symmetry means that $f(\bx)=f(\bx_\sigma)$ for any permutation $\sigma$
of the variables, and $f$ is a formal power series.
The ring $\Lambda_n$ of homogeneous symmetric functions of degree $n$ has several important bases. 
The \defn{homogeneous symmetric functions} \. $h_\la$ are defined as 
$$h_m(x_1,\ldots) = \sum_{i_1\leq i_2\leq \cdots \leq i_m} x_{i_1}\cdots x_{i_m} \qquad h_\la :=h_{\la_1} h_{\la_2}\cdots,$$
\defn{elementary symmetric functions} \. $e_\la$ given by
$$e_m(x_1,\ldots) = \sum_{i_1<  i_2 < \cdots < i_m} x_{i_1}\cdots x_{i_m} \qquad e_\la :=e_{\la_1} e_{\la_2}\cdots,$$
\defn{monomial symmetric functions} \. $m_\la := \sum_{\sigma} x_{\sigma_1}^{\la_1} x_{\sigma_2}^{\la_2}\cdots$ summing over all permutations of the indices giving distinct monomials,
\defn{power sum symmetric functions} \. $p_{\la}$ given by
$$p_m(x_1,\ldots) = \sum_{i} x_{i}^m  \qquad p_\la :=p_{\la_1} p_{\la_2}\cdots.$$
The \defn{Schur functions} \. $s_\la$ can be defined as the generating functions for SSYTs of shape $\la$
$$s_\la = \sum_{\mu \vdash n} K_{\la\mu} m_\mu.$$
The ring $\Lambda[\bx]$ is equipped with an inner product $\langle \cdot, \cdot \rangle$, where $\langle s_\la, s_\mu \rangle = \delta_{\la,\mu}$ and $\langle h_\la, m_\mu \rangle=\delta_{\la,\mu}$. 
The Schur functions at  can also be computed  using \defn{Weyl's determinantal} formula when the number of variables is finate:
$$s_\la(x_1,\ldots,x_\ell) = \frac{ \det [x_{i}^{\la_j+\ell-j}]_{i,j=1}^\ell}{\prod_{i<j} (x_i-x_j)}=\frac{a_{\la+\delta(\ell)}(x_1,\ldots,x_\ell)}{\Delta(x_1,\ldots,x_\ell)},$$
where $\Delta(x_1,\ldots,x_\ell) := \prod_{i<j\leq \ell} (x_i-x_j)$ and $a_{\alpha}(x_1,\ldots,x_\ell) = \det[ x_i^{\al_j}]_{i,j=1}^\ell$ are the alternants with $\delta(\ell) = (\ell-1,\ldots,1,0)$.  
The \defn{Jacobi--Trudi} identity gives
$$s_\la = \det [h_{\la_i -i+j}]_{i,j=1}^{\ell(\la)}.$$
The Schur function $s_\la(x_1,\ldots,x_N)$ is the character of $V_\la$ evaluated at a matrix with eigenvalues $(x_1,\ldots,x_N)$. 

\subsection{Multiplicities}\label{ss:mult}

The representation theoretic multiplicities defined earlier can all be expressed in terms of coefficients in the expansions of symmetric functions. Namely, we have that the Kostka numbers satisfy:
\begin{align}
s_\la = \sum_{\mu \vdash n} K_{\la,\mu} m_\mu\, \; \qquad h_{\mu} = \sum_\la K_{\la,\mu} s_\la.
\end{align}
The Littlewood--Richardson coefficients can be extracted as
\begin{align}
s_\mu(\bx) s_\nu(\bx) = \sum_{\la \vdash |\mu|+|\nu|} c^\la_{\mu\nu} s_\la(\bx) \, \qquad s_\la(\bx,\by) = \sum_{\mu,\nu} c^\la_{\mu\nu} s_\mu(\bx) s_\nu(\by) 
\end{align}
The Kronecker coefficients can be computed from
\begin{align}\label{eq:kron_def}
s_\la(x_1y_1,x_1y_2,\ldots,x_2y_1,x_2y_2,\ldots) = \sum_{\mu,\nu} g(\la,\mu,\nu) s_\mu(x) s_\nu(y).
\end{align}
Note that they satisfy properties like $g(\la,\mu,\nu) = g(\mu,\nu,\la) = g(\nu,\mu,\la)$ etc which are seen from their representation theoretic origins, see e.g.~\cite{S1}.
The plethysm coefficients are given via the plethysm of symmetric functions $f[g]=f(x^{\al^1},x^{\al^2},\cdots)$, where $g=x^{\al^1}+x^{\al^2}+\cdots$ is the expression of the function $g$ as a sum of monomials (possibly repeating) with $x^\alpha = x_1^{\alpha_1}x_2^{\alpha_2}\cdots$:
\begin{align}\label{eq:pleth_def}
s_\mu[s_\nu] = \sum_{\la \vdash |\mu| |\nu|} a^\la_{\mu,\nu} s_\la(\bx).
\end{align}
%

%

\subsection{Computational Complexity}\label{ss:cc}
We refer to \cite{Aa,Wig} for details on Computational Complexity classes, and to~\cite{Pak22,Pan23} and references therein for the connections with Algebraic Combinatorics. Here we recall the definitions of some of the classes mentioned in our discussion.

We will be conserned with computing functions $g(I)$ on some input vector $I$. We say that an algorithm computes  $g(I)$ in time $O(f(n))$ when $n$ is the size of the input $I$ (usually measured in bits)   and for every instance $I$ of such input size the algorithms takes at most $c f(n)$ many elementary steps where $c$ is some constant. We say that an algorithm solving a particular problem  runs in polynomial time, denoted $\poly(n)$, if there exists an integer $k$ independent of $n$ (but dependent on $g$), such that there is an algorithm computing $g(I)$ in time $O(n^k)$.  

A \defn{decision problem} is a computational problem, for which the output is Yes or No, interpreted at $1$ or $0$. There are two major complexity classes $\P$ and~$\NP$, subject of the $\P$ vs $\NP$ Millennium problem. $\P$ is the class of decision problems, where given any input of size $n$ (number of bits required to encode it), there is a fixed $k$, such that the answer can be obtained  time $\poly(n)$. $\NP$ is the class of decision problems, for which there exists a polynomial-time algorithm which verifies the Yes instances in polynomial time using a poly-sized witness.  Naturally, $\P \subset \NP$ and it is widely believed that $\P \neq \NP$.
%
%
%
The classes $\FP$ and $\SP$ are the counting analogues of $\P$ and $\NP$. A \defn{counting problem} is in $\FP$ if there is a $\poly(n)$ time algorithm computing the desired value. It is in $\SP$ if it is the number of accepting paths of a Turing machine solving an $\NP$ decision problem. In practice, $\SP$ is the class of problems searching to compute a function $g(I)$ on input $I$, which can be written as
$$ g(I) = \sum_{b \in \{0,1\}^{n^k} } M(b,I)$$
where $k$ is a constant, and $M(b,I)\in\{0,1\}$, $M \in \FP$.  That is, $\SP$ is the class of counting problems where the answers are exponentially large sums of $0$-$1$ functions, each of which can be computed in $O(\poly(n))$ time. $\SP$ is closed under addition and multiplication. Closing $\SP$ under subtraction we obtain $\GapP = \{ f-g |\, f,g \in \SP\}$ and we define the subset of its positive functions as $\GapP_{\geq 0} = \{ f-g | \, f,g \in \SP \text{ and } f-g \geq 0\}$. Naturally $\FP \subset \SP \subset \GapP_{\geq 0} \subset \GapP$. 

\subsection{Useful inequalities and notation}

We will use the following simple inequalities
\begin{align}\label{ineq:binom1}
\binom{a}{b} \geq \left(\frac{a}{b}\right)^b \text{ for $a \geq b$}, \quad \text{ and} \qquad \binom{a+b-1}{b} \leq a^b \text { for all $a,b \geq 0$}.
\end{align}
We denote by $\log$ the logarithm with base $2$, so $\log(2)=1$, and by $\ln$ the natural logarithm. We will use Stirling's approximation
\begin{align}\label{ineq:stirling}
n! \sim \sqrt{2\pi} n^{n+1/2} e^{-n}, 
\end{align}
which comes with very tight bounds. We also have that 
$$\sum_{\la \vdash n} (f^\la)^2 =n!,$$
which immediately implies $f^\la \leq \sqrt{n!}$, and $\max f^\la \geq \sqrt{n!}e^{-\pi\sqrt{2/3} \sqrt{n}}$ from bounding the number of integer partitions via Hardy-Ramanujan. 

\section{Dimension growth}\label{sec:dimension}

Here we investigate the asymptotic behavior of $f^\la$ in various regimes. Our goal is to identify when the runtime bounds from the quantum algorithms of \cite{LH24} are actually polynomial. 

There are three general regimes of growth -- polynomial $O(n^k)$ for fixed $k$, exponential $O(e^{cn})$, and superexponential $O(e^{c n \log n})$. 
While these have been studied in the literature in various contexts, here we will rederive and classify families of partitions exhibiting the above orders of growth for their dimension. We will use as a measure the Durfee square $d(\la)$, i.e.\ the size of the diagonal, of $\lambda$ and the  aft of $\la$ defined as ${\rm aft}(\la):=n-\max\{\la_1,\ell(\la)\}$ and show the classification in Table~\ref{table:growth_order}.
\begin{table}[!h]
\begin{tabular}{|p{1in}|p{1.5in}|p{1.5in}|p{1.5in}|}
\hline 
 & ${\rm aft}(\la)=k$ & $d(\la)=k$ & $d(\la) = \lfloor c \sqrt{n} \rfloor$ \\
\hline
$f^\la\leq $ &  $n^{k}$ & $(2k)^n $ & $c_1^{ n \log(n)}$ \\
\hline
$f^\la \geq$ & $ \binom{n-k}{k} $ & $b^n$, for some $b$ such that $\la_1,\ell(\la) \leq n/b$  & $c_2^{n\log n}$\\
\hline
\end{tabular}
\caption{Order of growth of $f^\la$ depending on the characteristics of the shape.} \label{table:growth_order}
\end{table}

We will now proceed to proving the above classification in several Propositions. The asymptotic behavior of $f^\la$ depends on the regime and how $\la$ grows more precisely, to obtain a leading order asymptotics with the leading constant $\la$'s growth would need to be specified as a limit shape (e.g. how each part $\la_i$ grows with $n$ exactly), see the Examples in this  Section. 

\begin{prop}\label{prop:fixed_aft2}
If ${\rm aft}(\la) =k$ then $f^\la \leq n^k/\sqrt{k!}$ and $f^\la \geq \binom{n-k}{k}$.
\end{prop}
\begin{proof}
Since $f^\la = f^{\la'}$, i.e. the number of SYTs is the same if $\lambda$ is transposed, we can assume that $\la_1 \geq \la'_1=\ell(\la)$.
To obtain the upper bound note that to get an SYT of shape $\la$ we need to choose the $n-k$ entries for the first row, order them in increasing order, and with the rest create an SYT of shape $\mu:=(\la_2,\la_3,\ldots)$. This is an overcount as the ``stitching'' of the two tableaux might violate the increasing columns conditions, so $f^\la \leq \binom{n}{k} f^\mu$. Since $\mu \vdash k$ we have $f^\mu \leq \sqrt{k!}$. As $\binom{n}{k} \leq n^k/k!$ the upper bound follows.

For the lower bound we can create SYTs of shape $\la$ by setting the first $\la_2$ entries of the first row  of $\la$ to be $1,2,\ldots,\la_2$, then selecting the entries in the rest of the first row in $\binom{n-\la_2}{k} \geq \binom{n-k}{k}$ ways and arranging the remaining entries in the remaining tableau of shape $\mu=(\la_2,\la_3,\ldots)$ in $f^\mu$ many ways.  
\end{proof}

The above lower bound is not so good in general. E.g. suppose that $k=n/2$, then the lower bound becomes just the trivial 1. Hence we need a more detailed approach towards understanding when the asymptotic behavior of $f^\la$ is polynomial.

\begin{prop}\label{prop:fixed_aft}
Suppose that $\la\vdash n$ is such that $f^\la \leq n^k $ for a fixed integer $k$ and assume that $n$ is large enough\footnote{for example, take $n>2^{10}k^4$}. Then $\max\{ \la_1 , \ell(\la)\}> n-4k^2$, so we have ${\rm aft}(\la) \leq 4k^2$. 
\end{prop}

\begin{proof}
Assume that $\la_1 \geq \ell(\la)$ by symmetry of $f^\la$ with respect to transposition of $\la$, we will first show that $\ell(\la)\leq 2k$.  Write $\la=(a_1,a_2,\ldots \mid b_1,b_2,\ldots)$ in Frobenius coordinates, that is $a_i = \la_i-i$ and $b_i = \la'_i-i$. Considering the combinatorial definition for $f^\la$ as counting the number of SYTs of shape $\la$, we have that $f^\la$ is bounded below by the product of $f^{\theta^i}$ where $\theta^i = (1+a_i,1^{b_i})$ are the principal hooks of $\la$, since we can just make an SYT from the SYTs for $\theta^1,\ldots$ by shifting the entries in $\theta^i$ by $|\theta^1|+\cdots+|\theta^{i-1}|$. Let $c_i=\min\{a_i,b_i\}$. We have that $f^{\theta^i} = \binom{a_i+b_i}{b_i} \geq 2^{c_i}$ after applying \eqref{ineq:binom1}. Thus 
$$n^k \geq f^\la \geq 2^{c_1+c_2+\cdots}. $$
Consequently, $\ell(\la)-1=c_1 \leq k\log(n)<\sqrt{n}/e$, where the second inequality holds for sufficiently large $n$ (e.g. $n>k^3$). 
Consider again $f^{\theta^1}$, we have
$$n^k \geq \binom{ \la_1 + c_1-1}{c_1} \geq \left(\frac{ \la_1+c_1}{c_1}\right)^{c_1} \geq \left( \frac{n/c_1 +c_1}{c_1 }\right)^{c_1}>\left(\frac{n}{c_1^2}\right)^{c_1}.$$
The function $g(x)= (n/x^2)^x$ is increasing for $x<\sqrt{n}/e$, which encompasses the interval $c_1$ is in. For $x=2k$ we have $g(2k)>n^k$ already for large enough $n$ ($n>16k^2$), and thus we must have $c_1<2k$
 and in particular $\ell(\la)\leq 2k$.

Next, suppose that $\la_2 =m$, and so $(m,m) \subset \la$ and $f^{(m,m)}\leq f^\la$. We have that $f^{(m,m)}=C_m = \frac{1}{m+1}\binom{2m}{m}$ (Catalan number). By the well-known asymptotics $C_m \sim 4^m/(m^{3/2} \pi)$ we see that $C_m \geq 2^m$ for $m\geq 5$. So $n^k \geq f^\la \geq f^{(m,m)} \geq 2^m$ and $m\leq k\log(n)$. 
Since $\la_1 \geq n-(2k-1)\la_2$, we get $\la_1 \geq n - (2k-1)m$ and $\la_1-\la_2 \geq n-2km$. Then $f^\la \geq \binom{ \la_1}{\la_2} \geq \binom{ n-2km}{m}$, since we can create an SYT by putting $1,\ldots,m$ in the beginning of the first row, choose $\la_1-m$ numbers from $m+1,\ldots,\la_1+m$ to be in the rest of the first row, and arrange the remaining numbers in the lower rows to create an SYT. 
Then $n^k \geq \binom{n-2km}{m} \geq \left( \frac{n-2km}{m} \right)^m$, and thus
$$k \log(n) \geq m\log(n) +m \log\left(1 - \frac{2km}{n}\right) -m\log(m) \geq m\log(n) - 2m\log(m),$$
where the last inequality follows since for $2mk<n$ we have $\log(1-2km/n) \geq -\log(m)$ for $m \geq k\log(n)$ and $n$ large enough. Finally, since $m^4< k^4 \log(n)^4< n$ for $n$ large enough, so $4\log(m) \leq \log(n)$, we use it in the last line $k\log(n) \geq m\log(n) -2m \log(m)\geq  m\log(n) - 2m \log(n)/4$ to get  $m\leq 2k$.   This gives $\la_2 + \cdots \leq 2k m \leq 4k^2$ and the result follows. 
\end{proof}

We now consider the case of a fixed Durfee size, i.e. $d(\la) =k$ and invoke the result of~\cite{regev}, which states 
\begin{thm}[\cite{regev}]
Let $d(\la) \leq k$, then $f^\la \leq (2k)^n$. 
\end{thm}

As a counterpart and lower bound to this result we invoke~\cite{GM}, slightly rephrased here.

\begin{thm}[\cite{GM}]
Suppose that $\la$ is such that $\la_1, \ell(\la) \leq n/a$ for some $a>1$. Then there exists a real number $b\in (1,a)$, such that $f^\la \geq b^n$. 
\end{thm} 
We immediately have that $\la_1,\ell(\la) <n/b$ also, although this should not be the tight bound. 

If we assume that $\la$'s Frobenius coordinates are $(n/\alpha_1 ,  n/\alpha_2 ,\ldots \mid n/ \beta_1 ,  n/\beta_2,\ldots)$, i.e. $\la_i -i = n/\alpha_i$  and $\la'_i -i=n/\beta_i $ for $i \leq d(\la)$,  then we invoke the results of \cite{MPPa} stating that
\begin{equation}\label{eq:tvk}
\ln f^\la = \sum_i\left( \ln( \alpha_i)/ \alpha_i + \ln(\beta_i)/\beta_i\right)n +o(n). 
\end{equation}
 
 We now consider the case of $d:=d(\la)$ not being constant. Let $\mu = (d^d)$ be the square inside $\la$. Then $f^\la \geq f^\mu$. Using the hook-length formula and approximating it via a Riemann integral we have that 
 $$\ln f^\mu = d^2 \ln(d) + (\ln(4)-3/2)d^2 + o(d^2).$$
 In particular if $d>c \sqrt{n}$, then $\ln f^\la = O( n\ln n)$. Modifying the underlying constants we can thus bound $f^\la$ by $f^\mu$ below and by $\sqrt{n!}$ above obtaining the following, see also~\cite{PPY}.
 \begin{prop}
 Suppose that $d(\la)/\sqrt{n}>c$ for some $c>0$. Then there exist constants $c_1,c_2$, such that 
 $$c_2^{n \log n} \leq f^\la \leq c_1^{n \log n}.$$ 
 \end{prop}  
 
 \medskip
 
Next we consider the asymptotics of $\frac{f^\la}{f^\mu f^\nu}$ in the various settings of \cite{LH24}. 

First, suppose that $\la \vdash n, \mu \vdash m, \nu \vdash n-m$ and let $\mu, \nu \subset \la$, this is the case of interest for computing $c^\la_{\mu\nu}$.  If $\ell(\la) =k$ is fixed, then $f^\la=O(c^n)$ and thus $\frac{f^\la}{f^\mu f^\nu} = O(e^{\log(c) n} )$ can be at most exponential. Could it be superexponential? The trivial answer is yes, since we can take $\la$ to be maximal, so $f^\la =O(\sqrt{n!})$ and $\mu=(m), \nu=(n-m)$ have dimensions 1. But then $c^\la_{\mu\nu}=0$ trivially too. 

\begin{prop}
Suppose that $\la \vdash n$, $\mu \vdash m$, $\nu\vdash n-m$ and $c^{\la}_{\mu\nu}>0$. Then $\frac{f^\la}{f^\mu f^\nu}  = O(2^n)$. 
\end{prop}
\begin{proof}
Using the general approach from~\cite{PPY} of expanding using symmetric functions and comparing coefficients we start with the defining identity for the LR coefficients:
$$s_\mu(\bx) s_\nu(\bx) = \sum_\la c^\la_{\mu\nu} s_\la(\bx).$$
Expanding the above in terms of monomials in $\bx$, we consider the coefficient of $x_1\cdots x_n$ on both sides. The coefficient at $x_1\cdots x_n$ at $s_\la(\bx)$ is equal to $K_{\la,1^n}=f^\la$. To get a monomial $x_1\cdots x_n$ at $s_\mu(\bx) s_\nu(\bx)$ we need to get a monomial $x_{i_1}\cdots x_{i_m}$ from $s_\mu$ for some $i_1<\cdots <i_m$ and the remaining $x$ variables should form a monomial $x_{j_1}\cdots x_{j_{n-m}}$ which we get from $s_\nu$. The coefficient at $x_{i_1}\cdots x_{i_m}$ at $s_\mu(\bx)$ is $f^\mu=K_{\mu,1^m}$ for $i_1<i_2<\cdots <i_m$, similarly the coefficient at $x_{j_1}\cdots x_{j_{n-m}} $ at $s_{\nu}(\bx)$ for $j_1<\cdots <j_{n-m}$ is $f^\nu$. The number of ways to choose the indices $i_1<\cdots <i_m$ among  $1,\ldots,n$ is $\binom{n}{m}$. Equating both sides then gives 
$$\binom{n}{m} f^\mu f^\nu = \sum_\la c^\la_{\mu\nu} f^\la,$$
and so if $c^\la_{\mu\nu}>0$ then $\frac{f^\la}{f^\mu f^\nu} \leq \binom{n}{m} \leq 2^n$. 
\end{proof}
Of course, this is only an upper bound. We would be interested to see when the ratio is of order $poly(n)$. Characterizing this completely is beyond the current technology, however we can exhibit families of partitions which could give polynomial growth for that ratio.

\begin{ex}
Suppose that $\la = (n-a,1^a)$, $\mu =(m-b, 1^b)$ and $\nu = (n-m-c,1^c)$ for some integers $a,b,c$. If $m,a,b,c$ grow proportional to $n$ then Stirling's approximation gives 
\begin{align*}
\frac{f^\la}{f^\mu f^\nu} = \frac{ \binom{n-1}{a} }{\binom{m-1}{b} \binom{n-m-1}{c}} \sim \sqrt{  \frac{ 2\pi (n-1) bc(m-1-b)(n-m-1-c)}{a(n-1-a) (m-1)(n-m-1)} } \\ \frac{ \left(\frac{b}{m-1}\right)^b \left( 1-\frac{b}{m-1} \right)^{m-1-b} \left( \frac{c}{n-m-1}\right)^c \left(1-\frac{c}{n-m-1}\right)^{n-m-1-c} }{ (a/(n-1))^a (1-a/(n-1))^{n-1-a} }.
\end{align*}
The above ratio would still be exponential in general. For example, if  $b\sim 1/3 (m-1)$, $c \sim 1/3(n-m-1)$ and $a \sim 1/2 (n-1)$, then the ratio simplifies to  
$ 2 (2^{5/3} /3 )^{n-2}$. However in that case $c^\la_{\mu\nu}=0$. 
For the case of a hook, we have that $c^\la_{\mu\nu}>0$ iff $a-b=c$ or $a-b=c+1$ by the Littlewood--Richardson rule.  
\end{ex} 

\begin{ex}
More generally, suppose that $\la,\mu,\nu$ are in the Thoma--Vershik--Kerov shape limit considered in~\cite{MPPa}, that is when the diagonal of the partition is fixed (or grows at a slower rate) and the lengths of the rows above the diagonal and the columns below the diagonal grow. Set $\la =(n/\al_1,\ldots \mid n/\be_1, \ldots)$ in Frobenius coordinates, $\mu = (m/\gamma_1,\ldots \mid m/\theta_1,\ldots)$, $\nu = ((n-m)/\pi_1,\ldots \mid (n-m)\rho_1,\ldots)$. Let $m=r n$ (treat as $m \approx rn$) for some real number $r<1$. Then by \eqref{eq:tvk} we have
$$\ln\left(\frac{f^\la}{f^\mu f^\nu} \right) = \sum_i  \left( \frac{\ln(\al_i)}{\al_i} + \frac{\ln(\be_i)}{\be_i} - r \frac{\ln(\ga_i)}{\ga_i} - r\frac{ \ln(\theta_i)}{\theta_i} - (1-r) \frac{\ln(\pi_i)}{\pi_i} -(1-r) \frac{\ln(\rho_i)}{\rho_i} \right)n + o(n),$$
and the ratio could be polynomial only if the linear factor vanishes.  This happens for example when $\al_i=\ga_i=\pi_i$ and $\be_i=\theta_i=\rho_i$ for all $i$, i.e. $\mu$ and $\nu$ are proportional to $\la$. 
\end{ex}

\begin{ex}
A more interesting regime is when $f^\la$ is superexponential, e.g. $\la$ has the Vershik--Kerov--Loggan--Shepp shape (see~\cite{MPPa} for the definition and a picture) which is equivalent to $f^\la \sim \sqrt{n!}$. Let $\mu \vdash n-k$ for some fixed $k$ and suppose $\mu$ is also of such shape, so $f^\mu \sim  \sqrt{m!}$. Then
$\frac{f^\la}{f^\mu} = \sqrt{ n(n-1)\ldots (n-k+1)} \leq n^{k/2}$ is polynomial and so would be $f^\la/(f^\mu f^\nu)$. However, in that case $c^\la_{\mu\nu}$ can be computed in polynomial, in fact -- constant time, as the number of possible LR tableaux of shape $\la/\mu$ and type $\nu$ is constant, depending only on the fixed $k$.
\end{ex}

\medskip

We now consider the setting of the Kronecker coefficients when $\la,\mu,\nu \vdash n$. Suppose $f^\la \geq f^\mu \geq f^\nu$. If $g(\la,\mu,\nu)>0$ then we must have $1 \leq \frac{f^\mu f^\nu}{f^\la}\leq f^\nu$. The upper bound for this quantity is easily seen to be $O(\sqrt{n!})$ as all three partitions can be large. Clearly, if $f^\nu =O(\poly(n))$ then the ratio is also polynomial. The converse does not need to hold. 

\begin{ex}
Suppose that $\la,\mu,\nu$ are of the Thoma-Vershik-Kerov shape and have Frobenius coordinates $\la=(n/\al_1,\ldots \mid n/\be_1,\ldots)$, $\mu=(n/\ga_1,\ldots \mid n/\theta_1,\ldots)$, $\nu = n/\pi_1,\ldots \mid n/\rho_i, \ldots)$. Then
$$\log\left(\frac{f^\mu f^\nu}{f^\la} \right) = \left( \sum_i -\frac{\log(\al_i)}{\al_i} - \frac{\log(\be_i)}{\be_i}+ \frac{\log(\ga_i)}{\ga_i} + \frac{ \log(\theta_i)}{\theta_i} + \frac{\log(\pi_i)}{\pi_i} + \frac{\log(\rho_i)}{\rho_i} \right)n + o(n), $$
and we can have subexponential growth as long as the factor at n is 0. For example, let $\mu=\nu=(x n,(1-x)n)$ and $\la =(n/2,n/2)$. Solving $-x\log(x) -(1-x)\log(1-x) = log(2)/2$ we get $x\approx 0.8899$ and the exponential terms disappear.  We can extend this to double hooks (the shape when $d(la)=2$ and the first two rows, first two columns grow linearly in $n$, in Frobenius coordinates this is $(a_1n,a_2n|b_1n,b_2)$) by reflecting the partitions about their diagonals. This makes the first case of triples $(\la,\mu,\nu)$ which is not covered by the general poly-time algorithms. In this case the Kronecker coefficient is polynomially sized by~\cite{PP22} and we believe can also be computed in polynomial time by analyzing the approach obtaining the asymptotics. 
\end{ex}

\begin{ex}\label{ex:rect}
It would be good to exhibit polynomial growth when the partitions do not have fixed diagonal length. Here we will argue that this is possible to achieve without giving the exact shapes due to number theoretic issues. Let $\mu=\nu=(\ell^m)$ where $m=\lfloor \ell^{r} \rfloor$ for some real number $r$ and so $n \sim \ell^{r+1}$. Using the hook length formula, for a partition $\al=(a^b)$ with $ab=n$ we have the leading term asymptotics
\begin{align*}
\log(f^\al) &= n\log(n) -n - \int_0^a \int_0^b \log(x+y) dx dy +o(n) \notag \\
&=   
n\log(n) -n - \frac12(a+b)^2 \log(a+b) +2ab +\frac12 a^2\log(a) +\frac12 b^2\log(b) +o(n)  \\
&= ab(\log(a)+\log(b)-\log(a+b)) -\frac12a^2 \log(1+b/a) - \frac12b^2\log(1+a/b) +ab +o(n) \notag
\end{align*}
If $a,b\sim \sqrt{n}$ then the leading term above is $\frac12 n \log n$. Now let $a=\ell$, $b=\ell^r$, and $n=\ell^{r+1}$ with $r>1$, we have
\begin{align*}
n\left(\log(\ell)+r\log(\ell)-\log(\ell) -\log(1+\ell^{r-1})\right) - \frac12 \ell^2 \log(1+\ell^{r-1}) -\frac12 \ell^{2r} \log(1+\ell^{1-r}) +n +o(n) \\
=n(r\log(\ell) - (r-1)\log(\ell) - O(\ell^{1-r}) ) - \frac12 \ell^2 ((r-1)\log(\ell)+O(\ell^{1-r}))\\ - \frac12 \ell^{2r} (\log(\ell^{1-r}) -O(\ell^{2(1-r)}) ) +n + o(n) 
= n\log(\ell) +O(\ell^2 \log(\ell) ) = \frac{1}{1+r} n\log(n) + o(n \log(n) )
\end{align*}
Thus if $\mu=(\ell_1^{\ell_1^{r_1}})$, $\nu=(\ell_2^{\ell_2^{r_2}})$ and $\la =(\ell^{\ell^r})$, such that $\frac{1}{1+r_1} + \frac{1}{1+r_2} = \frac{1}{1+r}$, then
$$\log\left( \frac{f^\mu f^\nu}{f^\la} \right) = o(n \log(n) ).$$
We'd expect that if $\la$ is close to that rectangle but with more rugged right boundary the ratio would decrease and possibly become polynomial. It is worth noting though that in this case $\ell(\mu)\ell(\nu) = \ell(\la)$ and by standard arguments (see e.g. \cite{IP24}) the rectangular parts can be removed, which would lead to computing the Kronecker coefficient for much smaller partitions. 
\end{ex}

\section{Computing Kostka and Littlewood--Richardson coefficients}\label{sec:kostka}

Here we extend some of the results in~\cite{LH24} on classical algorithms computing Kostka and Littlewood--Richardson coefficients. In~\cite{LH24} a classical algorithm computing $K_{\la\mu}$ is given which runs in time $O(f^\la)$. One way to do this is by generating all SYTs, which can be done dynamically by labeling each possible corner $u$ of $\la$ by $n$ and then proceeding to generate $\la \setminus u$. Then for each SYT we can check if it is the standardization of an SSYT of type $\mu$. That means that the numbers $\mu_1+\cdots+\mu_i+1,\ldots,\mu_1+\cdots+\mu_{i+1}$ appear in this order from left to right in the tableaux (no two in same column), so when we replace them by $i+1$ we will obtain a valid SSYT of type $\mu$. For example if $\la=(4,3,2)$ and $\mu=(3,3,3)$, then
$$\ytableaushort{1236,459,78} \quad \to \quad \ytableaushort{1112,223,33}$$
but the tableau $\ytableaushort{1234,567,89}$ does not standardize to an SSYT of type $\mu=(3,3,3)$ because $4$ appears to right of $5,6$. We then add to the counter 1 only when the SYT standardizes to an SSYT of shape $\mu$. 

However, as we saw in Section~\ref{sec:dimension}, $f^\la$ is very rarely of order $\poly(n)$ and this algorithm is not really efficient. A better algorithm is to use the correspondence between SSYTs and Gelfand-Tsetlin patterns. We consider $K_{\la/\mu,\nu}$ more generally, i.e. the number of SSYTs of type $\nu$ and skew shape $\la/\mu$. Let $ \ell(\mu)=b,\ell(\nu)=c$, $a=b+c$ and assume $\ell(\la)\leq a$ (otherwise $K_{\la/\mu,\nu}=0$ anyway). The corresponding Gelfand-Tsetlin polytope $GT(\la/\mu;\nu)$ consists of the points $\{(x_{ij})\in \mathbb{R}^{ a(a-1)/2-b(b-1)/2}, i=1,\ldots,c+1; j= 1,\ldots, i+b-1\}$ such that 
\begin{align}
  x_{i,j} \geq x_{i,j+1} \, ,& \qquad x_{i,j} \leq x_{i+1,j}  \, , \qquad x_{i,j} \geq x_{i+1,j+1} \\
 x_{1,j} = \mu_j \, ,& \qquad x_{a,j}=\la_j \quad \text{ for all $j$} \\
 \sum_j x_{i+1,j} -& \sum_j x_{i,j} = \nu_i \quad \text{for all $i$}. 
\end{align}
Then $K_{\la/\mu,\nu}$ is equal to the number of integer points in this polytope. Using  Barvinok's algorithm for counting integer points in polytopes of fixed dimensions~\cite{barvinok1994polynomial} we have a poly-time algorithm for $K_{\la/\mu,\nu}$ whenever $\ell(\la), \ell(\nu)$ are fixed.

\begin{rem}\label{rem:kostka_poly}
This efficient algorithm does not apply when the lengths are not constant. In an earlier version of this paper we posed  a conjecture generalizing the result in~\cite{LH24}. 
We conjectured that there exists an algorithm computing $K_{\la/\mu,\nu}$ in time $O(K_{\la/\mu,\nu} \poly(n))$. 
The difficulty to creating such an algorithm lies in the recursive generation of SSYTs. While removing corners to generate SYTs always results in a valid SYT, removing horizontal strips would also generate valid SSYTs, however ensuring that they are of the desired type $\nu$ would often lead to dead ends. Thus such an algorithm would run longer than the number of actually desired SSYTs.
There is one caveat however. In most cases when $\la,\mu,\nu$ don't scale linearly we have that $K_{\la/\mu,\nu}$ is exponentially (or superexponentially) large, see~\cite{PP20,Pan23}. All these coefficients, and multiplicities in question, are easily seen to be computable in exponential time, see the discussion in Sections~\ref{sec:rem_kron} and~\ref{sec:rem_pleth}. So the conjecture is vacuously true in most cases. The interesting cases arise when $K_{\la/\mu,\nu}$ is of polynomial size and $\ell(\la)$ is not fixed. 

However, skew Kostka numbers are special cases of Littlewood--Richardson coefficients $c^\la_{\mu,\nu}$ by a polynomial substitution.  Namely $K_{\la/\mu,\nu} = c^{\theta}_{\al\nu}$ where $\theta_i = \la_1 + \nu_1+\cdots+\nu_{\ell(\nu)+1-i}$ for $i=1,\ldots, \ell(\nu)$ and $\theta_{i+\ell(\nu)} = \la_i$, $\al= (\la_1^\ell(\nu), \mu)$, see e.g.~\cite{Pan23} and the discussion in~\cite{StaMO}.
The Littlewood--Richardson coefficients allow for algorithms with runtime polynomial in the size of the coefficient as discussed below. 
\end{rem}

\smallskip

We now proceed with computation of LR coefficients $c^\la_{\mu\nu}$. As with Kostka, they count integer points in a polytope, in this case it is the LR polytope (see~\cite{PV05}) which can seen as a section of the GT polytope $GT(\la/\mu;\nu)$. Set the skew shape $\theta=\la/\mu$, so that the number of boxes on row $i$ is $\theta_i=\la_i-\mu_i$ and let $k=\ell(\theta)$ (i.e.\ the number of nonzero entries in $(\theta_1,\theta_2,\ldots)$) and in particular $k \leq \ell(\la)$. The polytope $LR(\la/\mu;\nu)$  is given by the set of points $\{a_{ij}\}$ with $a_{ij}\in \mathbb{Z}_{\geq 0}$ and $1\leq j \leq i \leq k$, satisfying the following inequalities
\begin{align}
\sum_{j\leq i} a_{ij} = \theta_i, &\text{ for } i=1,\ldots,k; \label{eq:LR_first} \\
\sum_{i=j}^k a_{ij} = \mu_j,& \text{ for }j=1,\ldots,k; \\
\nu_i+\sum_{j=1}^r a_{ij} \leq \nu_{i-1}+\sum_{j=1}^{r-1} a_{ij}, &\text{ for }i=2,\ldots,k, r=2,\ldots,i\\
\sum_{i=j}^r a_{ij} \leq \sum_{i=j-1}^{r-1} a_{i,j-1}, &\text{ for }j=2,\ldots,k, r=j,\ldots,k,
\end{align}
where $\theta=\la/\nu$ and $\theta_i = \la_i-\nu_i$. 
In particular $c^\la_{\mu\nu}$ can be computed in time $O( \log(|\theta|)^{ \ell(\theta)^2} )$ by Barvinok's algorithm and this gives $O(\poly(n))$ algorithm as long as $\ell(\theta)$ is fixed. 

We now consider other regimes depending on the shapes and sizes of the three partitions.

\begin{prop}\label{prop:LR_compute}
Suppose that $\la \vdash n$, $\mu \vdash k$ and $\nu \vdash n-k$. Then there exists an algorithm which runs in time $O(\log(\la_1)\ell(\la) + 2^k k^{k+2})$ which computes $c^{\la}_{\mu\nu}$.
\end{prop}
\begin{proof}
Consider the shape $\theta=\la/\nu$, whose rows have length at most $k$. We can remove all empty rows and columns of $\theta$, so that $\theta_i =\la_i-\nu_i \geq 1$ and the rows are connected at vertices of the grid as in $\ydiagram{3+2,3}$ or $\ydiagram{3+2,4}$, but not like $\ydiagram{3+2,2}$. We have $\sum_i \theta_i=k$ and the number of rows is at most $k$. 
We can now generate all Littlewood--Richardson tableaux as follows. 
We thus need to find all arrays of nonnegative integers satisfying the above conditions. While there are more efficient algorithms for counting integer points in polytopes and in particular Littlewood--Richardson coefficients, we can simply list all possible arrays satisfying the first set of conditions~\eqref{eq:LR_first} in $\prod_i \binom{ \theta_i+i-1}{i-1}\leq (2k)^k$ many ways, and check if it satisfies the other conditions in time $O(k^2)$.  
\end{proof}

\begin{prop}\label{prop:multiLR_compute}
Let $\la \vdash n$ and let $\al^i$ be partitions, such that $\al^1 \vdash n-k$, and $|\al^2|+\cdots+|\al^r|=k$ for some $k$. Define the multi-LR coefficient $c^{\la}_{\al^1 \ldots \al^r} := \langle s_{\al^1}\cdots s_{\al^r},s_{\la} \rangle$, i.e.\ as the coefficient at $s_\la$ in the expansion of $s_{\al^1}\cdots s_{\al^r}$ in the Schur basis. Then the value of $c^{\la}_{\al^1\ldots \al^r}$ can be computed in time $O(\ell(\la_1) \log(\la_1) k^{2k^2})$.
\end{prop}
\begin{proof}
First we have the following recurrence relation
\begin{align}
c^{\la}_{\al^1\ldots \al^r} = \sum_{\mu \vdash k} c^{\la}_{\al^1 \mu} c^{\mu}_{\al^2 \ldots \al^r},
\end{align}
which follows from the recursion
$$s_{\al^1}\cdots s_{\al^r} = s_{\al^1} s_{\al^2}\cdots s_{\al^r} = s_{\al^1}\left( \sum_{\mu} c^{\mu}_{\al^2\ldots \al^r}s_\mu \right)=\sum_\mu c^{\mu}_{\al^2\ldots \al^r} s_{\al^1}s_\mu = \sum_\mu \sum_\la c^{\mu}_{\al^2\ldots \al^r} c^{\la}_{\al^1\mu} s_\la.$$
The number of partitions of $k$, denoted by $p(k)$, is smaller than $e^{\pi \sqrt{2/3} \sqrt{k}}$ by the famous Hardy-Ramanujan formula, and those partitions can be dynamically generated in time $p(k)$ via various recursions. For each such $\mu$ we compute $c^\la_{\mu \al^1}$ in time $O(\log(\la_1)\ell(\la) + 2^k k^{k+2})$ by Proposition~\ref{prop:LR_compute}. Next, let $a_i=|\al^i|$ and $b_i = a_{i+1}+\cdots$. Then we compute $c^{\mu}_{\al^2\cdots \al^r}$ iteratively as
$$c^{\mu}_{\al^2\cdots \al^r} = \sum_{\beta^i \vdash b_i, i=2,\ldots,r-1} c^{\mu}_{\al^2 \beta^2} c^{\beta^2}_{\al^3 \beta^3}\cdots c^{\beta^{r-1}}_{\beta^r \al^r}.$$
Generating the partitions $\beta^i$ would take $p(b_2)\cdots p(b_r) < e^{\pi \sqrt{2/3} \sqrt{k} r}$ many steps, and computing each LR coefficient would cost, by Proposition~\ref{prop:LR_compute} at most $O(2^k k^{k+2})$ where we bound each $a_i, b_i \leq k$. Thus the multi-LR would take at most $O(e^{\pi \sqrt{2/3} \sqrt{k}  r }2^{k(r-1)} k^{(r-1)(k+2)})$.  We have $r \leq k+1$ because the $\al^i$'s are nonempty so $r-1 \leq |\al^2|+\cdots|\al^r|=k$. Multiplying the runtimes and replacing $r$ by $k+1$ gives the upper  bound $O(\ell(\la_1)\log(\la_1)  e^{\pi \sqrt{2/3} k\sqrt{k}}  2^{k^2} k^{(k+1)(k+2)}  )$ which we can bound by the simpler $O(\ell(\la_1) \log(\la_1) k^{2k^2})$ (for $k\geq 5$). 
\end{proof}

\begin{prop}\label{prop:LR_kostka}
Let $\la \vdash n$, $\mu,\nu$ be partitions with $|\mu|+|\nu|=n$. Suppose there is an algorithm running in time $t$ which generates all SSYTs of shape $\la/\mu$ and type $\nu$. Then there exists a classical algorithm which computes $c^\la_{\mu\nu}$ in time $O(t n^2)$.
\end{prop}
\begin{proof}
The algorithm generates all SSYT of shape $\la/\mu$ and type $\nu$. 
The classical LR rule states that $c^\la_{\mu\nu}$ is equal to the number of SSYTs of shape $\la/\mu$, type $\nu$ and whose reading word is a lattice permutation. To obtain the reading word we record the entries of the tableaux starting from the top right corner, reading to the left entry by entry and moving to the next row. The resulting word should be a lattice permutation/ballot sequence, i.e.\ for every prefix and every $i$ the number of entries $i$ is $\geq$ the number of entries $i+1$ in that prefix.
For example, $\ytableaushort{ \none\none 111, \none 222,133}$ is an LR tableaux of shape $(5,4,3)/(2,1)$, type $(4,3,2)$ and the reading word is $111222331$.

To compute the LR coefficient we generate all SSYTs of shape $\la/\mu$, type $\mu$ in time $t$. 
For each such SSYT  we obtain the reading word in time $O(n)$ and then dynamically check if it is a ballot sequence in time $O(n^2)$. This subroutine can be optimized depending on the encoding (recording a tableaux by the number of entries equal to $j$ in row $i$) and improved to $O(\log(n) \log(\nu_1)\ell(\nu))$. 
\end{proof}

\begin{cor}
The LR coefficient $c^\la_{\mu\nu}$ can be computed in time $O(f^{\la/\mu}n^3)$.
\end{cor}

\begin{proof}
Generate all SYTs of shape $\la/\mu$ in time $O(f^{\la/\mu})$ by the recursive removal of corner boxes. For each such SYT, apply the standartization of type $\mu$ when possible in time $O(n)$. This generates the SSYTs of shape $\la/\mu$ in time $t=f^{\la/\mu}n$. Now apply Proposition~\ref{prop:LR_kostka}.
\end{proof}

In the same spirit as with the Kostka coefficients we reiterate the conjecture of~\cite{LH24}.

\begin{conj}[\cite{LH24}]\label{conj:lr}
Let $\la$, $\mu$, $\nu$ be partitions such that $|\la|=|\mu|+|\nu|$. Then there exists a classical algorithm running in time $O(\frac{f^\la}{f^\mu f^\nu} \poly(n))$ which computes $c^\la_{\mu\nu}$.
\end{conj}

In light of the dimension discussion in Section~\ref{sec:dimension}, we note that most cases of triples would give exponentially large ratios. It is easy to see that computing the LR coefficients would take at most exponential time (e.g. via characters). When $\ell(\la)$ is fixed we also have polynomial time algorithms, as it is equivalent to counting integer points in the fixed dimension hive polytope. Thus the interesting cases need the dimension ratio to be $\poly(n)$-large and $\ell(\lambda)$ growing, and we have yet to identify nontrivial such examples (as discussed in Section~\ref{sec:dimension}).

We now  pose a stronger conjecture which implies the previous one since $c^\la_{\mu\nu} \leq \frac{f^\la}{f^\mu f^\nu}$. 
\begin{conj}\label{conj:lr2}
Let $\la$, $\mu$, $\nu$ be partitions such that $|\la|=|\mu|+|\nu|$. Then there exists a classical algorithm running time $O(c^\la_{\mu\nu} \poly(n))$ which computes $c^\la_{\mu\nu}$ when $c^\la_{\mu,\nu}>0$.
\end{conj}

Again, as with the Kostka numbers, the interesting regime here is when $c^\la_{\mu\nu}$ is $\poly(n)$ large (as opposed to exponential, which is the usual case) and $\ell(\la)$ is not fixed. In almost all other cases the conjecture is vacuously true. 

As observed by Maxim van den Berg (personal communication), in most cases the conjectures above follow directly from the result of Ikenmeyer~\cite{ikenmeyer2016small}.

\begin{thm}[\cite{ikenmeyer2016small}]
The Littlewood--Richardson coefficient $c^\la_{\mu\nu}$ can be computed in time $O((c^\la_{\mu\nu})^2 \poly(n))$.
\end{thm}

In particular, if $c^\la_{\mu\nu}$ is of size $\poly(n)$ then the LR coefficient can be computed in polynomial time, which shows that Conjecture~\ref{conj:lr2} is true. This leaves a very small range for possible improvements where the LR are superpolynomial, but also subexponential.

\section{Kronecker coefficients}\label{sec:kron}

We now turn towards efficient algorithms for the Kronecker coefficients for triples $(\la,\mu,\nu)$ in various regimes.

\begin{prop}\label{prop:kron_k}
Let $\la,\mu,\nu \vdash n$ and set $k=n-\nu_1$, i.e.\ ${\rm aft}(\nu)=k$. Then
$g(\la,\mu,\nu)$ can be computed in time $O(( \ell(\la)\log(\la_1)+\log(\mu_1)\ell(\mu)) \min\{\ell(\la),\ell(\mu)\}^k  k^{2k^2+2k})$.
\end{prop}

\begin{proof}
Let $\ell:=\ell(\nu)$, then by the Jacobi--Trudy identity and the identification of $h_m=s_{(m)}$ -- the character of the trivial representation, we have
\begin{align*}
s_\nu[\bx\by] = \sum_{\sigma \in S_{\ell} } \sgn(\sigma) \prod_i h_{\nu_i+\sigma_i-i} [xy] \\
= \sum_{\sigma \in S_{\ell} } \sgn(\sigma) \sum_{\al^i \vdash \nu_i+\sigma_i-i} \prod s_{\al^i}(x) s_{\al^i}(y). 
\end{align*} 
Here we used the notation $f[\bx\by]:=f(x_1y_1,x_1y_2,\ldots,x_2y_1,x_2y_2,\ldots)$, i.e.\ substituting all products of pairs as variables into $f$. We also used that $g( (m), \la,\mu)=1$ if $\la=\mu$ and 0 otherwise. 
Expanding both sides in terms $s_\al(x)s_\be(y)$, using~\eqref{eq:kron_def} for the left side, and comparing coefficients we have
\begin{align}
g(\la,\mu,\nu) = \sum_{\sigma \in S_{\ell} } \sgn(\sigma) \sum_{\al^i \vdash \nu_i+\sigma_i-i} c^\la_{\al^1,\ldots,\al^\ell} c^\mu_{\al^1,\ldots,\al^\ell}.
\end{align}
Note that in the formula above we have $\al^1 \vdash \nu_1 + \sigma_1-1 \geq n-k$, and so $|\al^2|+\cdots+|\al^\ell| \leq k$. Thus, for each choice of $\al^1,\ldots,\al^\ell$ we can compute the product of the two multi-LR coefficients in time $O( (\ell(\la)  \log(\la_1)+ \ell(\mu) \log(\mu_1) ) k^{2k^2})$ by Proposition~\ref{prop:multiLR_compute}. We can generate each of the $\al^2,\ldots,\al^\ell$ in $p(k)^{\ell-1}$ steps. Finally, we can choose $\al^1$  as follows. First, observe that we are only considering $\al^1 \subset \la\cap \mu$, otherwise the multi-LR coefficient would be 0. Assume that $\ell(\la)\leq \ell(\mu)$. So we can determine $\al^1$ via the values  $\la_i - \al^1_i \leq k$ for each $i=1,\ldots,\ell(\la)$.  For the nonzero such differences we choose a strong composition $(c_1,\ldots,c_r)$ of $k -\sigma_1+1$ in at most $\binom{k-1}{r-1}\leq k^r$ many steps. We choose a set $I=\{i_1,\ldots,i_r\} \subset \{1,\ldots,\ell(\la)\}$ of size $r$ in $\binom{\ell(\la)}{r} \leq \ell(\la)^r$ many ways and set $\al^1_{i_j}:=\la_{i_j}-c_j$ and $\al_i=\la_i$ for $i \not \in I$. We accept $\al^1$ if $\al^1_i\geq \al^1_{i+1}$ for all $i$ which can be checked in $\ell(\la)$ steps. Altogether generating $\al^1$ takes at most $O(\sum_{r=1}^k k^r \ell(\la)^r \ell(\la))=O(k^k\ell(\la)^k \ell(\la))$ steps.
Iterating over all $\ell!$ permutations $\sigma$ and repeating the above procedures gives the Kronecker coefficient.

The runtime of generating all $\al$ tuples involved  is thus $O(   \ell! k^k\ell(\la)^{k+1}p(k)^{\ell-1})=O(k^{2k}\ell(\la)^k)$, where we bound $\ell \leq k$ and $p(k) \sqrt{k}e^{-k} \leq 1$ from the asymptotic approximations~\eqref{eq:pn} and use Stirling's formula. Multiplying by the times it takes to compute each multi-LR gives the desired bound. 
\end{proof}

\begin{proof}[Proof of Theorem~\ref{thm:kron1}]
Suppose that $f^\nu \leq n^k$. By Proposition~\ref{prop:fixed_aft} we have that ${\rm aft}(\nu)\leq 4k^2$. Since $g(\la,\mu,\nu)=g(\la,\mu',\nu')$ we can assume $\ell(\nu)\leq \nu_1$, otherwise we repeat the argument with $\nu'$ and $\mu'$ instead. Set $m:=4k^2$, then $\nu_1 \geq n-m$ and we can apply Proposition~\ref{prop:kron_k}, so $g(\la,\mu,\nu)$ can be computed in time $ O(( \ell(\la_1)\log(\la_1)+\log(\mu_1)\ell(\mu_1)) \min\{\ell(\la),\ell(\mu)\}^m  m^{2m^2+2m})$. We have $\ell(\la_1) \leq n$, $\log(\la_1) \leq \log(n)$ and this gives the desired bound of $D(k) n^{4k^2+1}\log(n)$.
\end{proof}

\section{Plethysm coefficients}\label{sec:pleth}

We will start with one of the more ubiquitous\footnote{Ubiquitous because it is the most basic one and plays a special role in Geometric Complexity Theory, see e.g.~\cite{IP17}.} cases of plethysm coefficients: $s_{(d)} [s_{(m)}]$. If we are looking at the coefficient for $s_\la$, then it suffices to use only $\ell(\la)=k$ many variables. We have $s_{(d)}=h_d$, $s_{(m)}=h_m$. We expand 
$$h_d(y_1,\ldots) = \sum_r \sum_{ \substack{ c_1+c_2+\cdots+c_r =d \\ c_i>0}} \quad \sum_{i_1<i_2<\cdots<i_r} y_{i_1}^{c_1}\cdots y_{i_r}^{c_r}$$
and 
$$h_m(x_1,\ldots,x_k) = \sum_{ b_1+\cdots+b_k=m} x_1^{b_1}\cdots x_k^{b_k},$$
where $(b_1,\ldots,b_k)$ go over all compositions of $m$ allowing parts to be 0. Order these compositions in lexicographic order, so that $b<b'$ if there is an index $j$, such that $b_1=b_1',\ldots,b_{j-1}=b_{j-1}'$ and $b_j<b_j'$. This gives a total ordering of these compositions and we can list them as $b^1=(0,\ldots,0,m)<b^2=(0,\ldots,1,m-1)<b^3<\ldots$.  Expanding $h_m$ into a sum of monomials ordered by that lexicographic ordering and then substituting into the expression for $h_d$ above we get
\begin{align}\label{eq:pleth_expansion}
h_d[h_m(x_1,\ldots,x_k)] = \sum_{r\leq d} \sum_{ \substack{ c_1+\cdots+c_r =d\\ c_i>0} } \quad \sum_{i_1<i_2<\cdots<i_r} \bx^{c_1\cdot b^{i_1}}\cdots \bx^{c_r \cdot b^{i_r}},
\end{align}
where $\bx^{v}:=x_1^{v_1} x_2^{v_2}\cdots$ for any sequence $v=(v_1,v_2,\ldots)$.
Partition the space of possible $b^i$s into the following $k^{r-1}$ many polytopes for each $r=1,\ldots,d$. Let $\bar{j}:=(j_1,\ldots,j_r) \in [k]^{r-1}$ and define the set of points $P(\bar{j}) \subset \{ (b_l^i) \in \mathbb{R}^{rk}_{\geq 0}: i=1,\ldots,r; l=1,\ldots,k)$ as 
\begin{align}
\sum_{l=1}^k b_l^i =m, & \quad \text{ for } i=1,\ldots,r;\\
b^i_l =b^{i+1}_l, & \quad\text{ for }l=1,\ldots,j_i-1, \text{ for each } i=1,\ldots,r-1; \\
b^i_{j_i} < b^{i+1}_{j_i},  & \quad \text{ for }i=1,\ldots,r-1.
\end{align}
Then compositions $b^1, \ldots, b^r$ are ordered in increasing lexicographic order if and only if they belong to one of $P(\bar{j})$. Since $P(\bar{x}) \cap P(\bar{y})=\emptyset$ if $\bar{x} \neq \bar{y}$, we have that $P(\bar{j})$ partition the space of $b^1<b^2<\cdots<b^r$ as $\bar{j} \in [k]^{r-1}$. 
We can thus write
$$
h_d[h_m] = \sum_{r=1}^d \sum_{c: \substack{ c_1+\cdots+c_r=d\\ c_i>0 }} \sum_{\bar{j} \in [k]^{r-1}} \sum_{\bbb \in P(\bar{j})}   x^{c_1b^1+\cdots+c_rb^r},
$$
where the second sum is over strong compositions $c$ of $d$. 

We are now ready to expand into Schur functions. We have that $s_\la(x_1,\ldots,x_k) = \frac{ a_{\la+\delta(k)}(x_1,\ldots,x_k) }{\Delta(x_1,\ldots,x_k)}$, and  multiplying both sides of equation~\ref{eq:pleth_def} by $\Delta$, identifying $h_d=s_{(d)}$ and $h_{m}=s_{(m)}$, we obtain 
$$\Delta(x_1,\ldots,x_k) h_d[h_m(x_1,\ldots,x_k)] = \sum_{\la \vdash dm} a^\la_{d,m} a_{\la+\delta(k)}(x_1,\ldots,x_k).$$
We can extract the plethysm coefficient by extracting the monomial $\bx^{\la+\delta(k)}$ on each side (i.e.\ expanding each side into multivariate monomials and comparing the coefficients on both at the monomial $x_1^{\la_1+k-1}\cdots x_k^{\la_k}$). Denote by $[\bx^v]f(x_1,\ldots,x_k)$ the coefficient at the monomial $x_1^{v_1}\cdots x_k^{v_k}$ in the expansion of $f(x_1,\ldots,x_k)$ into monomials. , expanding $\Delta(x_1,\ldots,x_k) = \prod_{i<j}(x_i-x_j) = \sum_{\sigma \in S_k} \sgn(\sigma) \bx^{\sigma(\delta(k))}$, we get
\begin{align}\label{eq:pleth1}
a^\la_{d,m} &= [\bx^{\la+\delta(k)}] \Delta(x_1,\ldots,x_k) h_d[h_m(x_1,\ldots,x_k)] \notag \\
&= [\bx^{\la +\delta(k)}] \sum_{\sigma \in S_k} \sgn(\sigma) \bx^{\sigma(\delta(k))}  \sum_{r=1}^d \sum_{c: c_1+\cdots+c_r=d, c_i>0} \sum_{\bar{j} \in [k]^{r-1}} \sum_{b \in P(\bar{\bbj})}   \bx^{c_1b^1+\cdots+c_rb^r} \notag \\
&= \notag \sum_{\sigma \in S_k} \sgn(\sigma)  \sum_{r=1}^d \sum_{c: c_1+\cdots+c_r=d, c_i>0} \sum_{\bar{\bbj} \in [k]^{r-1}} \sum_{b \in P(\bar{j})}   \mathbf{1}[ c_1b^1+\cdots c_rb^r +\sigma(\delta(k)) = \la +\delta(k)] \\
&= \sum_{\sigma \in S_k} \sgn(\sigma)  \sum_{r=1}^d \sum_{c: c_1+\cdots+c_r=d, c_i>0}  \sum_{\bar{j} \in [k]^{r-1}}  |\{ b \in P(\bar{j}); c_1b^1+\cdots c_rb^r = \la +\delta(k) -\sigma(\delta(k)) \}|,
\end{align}
 where the last cardinality counts the number of integer points $b$ in $P(\bar{j})$.
We now define the polytope $Q(\bar{j}, c, \alpha)$ as the set of points $(b^i_j) \in P(\bar{j}) \cap \{ (b^i_j) \, | \, c_1b^1+\cdots +c_rb^r = \alpha\}$. So $Q(\bar{j},c,\alpha) \subset \mathbb{R}^{rk}$. Equation~\eqref{eq:pleth1} then gives
\begin{align}\label{eq:pleth_formula}
a^\la_{d,m} = \sum_{\sigma \in S_k} \sgn(\sigma) \sum_{r=1}^d \sum_{\substack{ c_1+\cdots+c_r=d\\ c_i>0}} \sum_{\bar{j} \in [k]^{r-1}} |Q(\bar{j},c,\la+\delta(k)-\sigma(\delta(k)))|
\end{align}

\begin{prop}\label{prop:pleth_gen}
Let $k,d,m$ be integers and $\lambda \vdash dm$ so that $\ell(\la)=k$. Then $a^\la_{d,m}$ can be computed in time $O(k! (2k)^{d-1} (dm)^{dk}\log(dk))$.  
\end{prop}
\begin{proof}
We apply formula~\eqref{eq:pleth_formula}. For each polytope $Q(\bar{j},c,\la+\delta(k) - \sigma(\delta))$ we have that the input bounds for the polytope are bounded by $m, \la_1+k \leq md$. The dimension of the polytope is at most $dk$ since $r \leq d$, and so by Barvinok's algorithm the number of integer points can be found in time $O((dm)^{dk}\log(dk))$. The total number of polytopes we consider  is the number of triples $(\sigma, c, \bar{j})$. The total number of compositions $c$ is $2^{d-1}$, there are  $k^{r-1} \leq k^{d-1}$ many vectors $\bar{j}$, and so the total number of polytopes is $\leq k! 2^{d-1} k^{d-1}$. Multiplying these bounds gives the total runtime. 
\end{proof}

The above formula gives an efficient algorithm when $k$ and $d$ are fixed.

The formulas can be used in different regimes. If ${\rm aft}(\la)=K$ is fixed, then as $n$ and $d$ grow we have that the plethysm coefficient stabilizes as soon as $d$ becomes larger than ${\rm aft}(\la)$, see e.g.~\cite{brion1993stable,colmenarejo2017stability}. This implies that $a^\la_{d,m} = a^\mu_{K,m}$ where $\mu = \la -(d-K)$, and then computing $a^\mu_{K,m}$ can be done in $\poly(n-d+K)$ time by Proposition~\ref{prop:pleth_gen}. For the sake of self-containment we provide a detailed proof.

\begin{prop}\label{prop:pleth_aft}
Let $\la \vdash n$ with ${\rm aft}(\la)=K$ with $\la_1 \geq \ell(\la)$ and $d,m$ be such that $dm=n$. Then $a^\la_{d,m}$ can be computed in time $O(n^{4K^3(K+1)}\log(dK))$.
\end{prop} 
\begin{proof}
First, if $d \leq 4K^3$ then the result follows from Proposition~\ref{prop:pleth_gen} since $\ell(\la) \leq K+1$. Now assume that $d >4K^3$. 

Consider formula~\ref{eq:pleth_formula} and, to avoid notational issues, let $k$ be the number of variables.  The set of points $(b^i_j)$ in the polytopes $Q(\bar{j},c,\lambda+\delta(k)-\sigma(\delta(k)))$ satisfy the equations 
\begin{align}\label{eq:pleth_constraints}
c_1b^1_j+\cdots +c_r b^r_j = \la_j+ k -j - \sigma(k -j).
\end{align}
Since $\ell(\la) \leq {\rm aft}(\la)+1 =K+1$, so $\la_j=0$ for $j>K+1$, we must have $\sigma(k-j)=k-j$ for all of these values, else there will be negative numbers. We can thus assume that $k = K+1$ and is fixed. Next, since $c_i >0$ in equation~\ref{eq:pleth_constraints} we need to have at most $\la_j + \ell-j-\sigma(\ell -j) <2K$ many nonzero terms for each $j\geq 2$, so $\#\{ i: b^i_j >0\} <2K$ for each $j\geq 2$. Thus the total number of nonzero entries in the $(b^i_j)$'s for $j \geq 2$ is at most $(k-1)2K \leq 2K^2$. Since the vectors $b^i$ are all compositions of $m$, and are supposed to be distinct, there is at most one vector with $b^i_j=0$ for $j\geq 2$ (that is, $(m,0,\ldots,0)$) and by the lexicographic ordering this should be the largest one, $b^r$.  Thus the total number of vectors distinct from $(m,0,\ldots)$ should be at most $2K^2$, and consequently $r \leq 2K^2+1$. Further, if a vector $b^i$ is such that $b^i_j >0$ for some $j \geq 2$, then since $c_i b^i_j \leq \la_j + k-j - \sigma(k-j) \leq 2K$, we must have $c_i \leq 2K$ for all $i$, such that $b^i_j >0$ for some $j\geq 2$. 

If we had that for all $i=1,\ldots,r$ we have $b^i_j >0$ for some $j\geq 2$  then 
$$d = \sum_i c_i \leq r 2K \leq 2K\cdot 2K^2,$$
which contradicts the assumption on $d$. We must thus have $b^r = (m,0,\ldots)$ and $c_r \geq d - 4K^3$. 
Formula~\eqref{eq:pleth_formula} can be rewritten with those constraints in mind, namely we can iterate over $r \leq 4K^3+1$ with $c_i \leq 2K$ for $i=1,\ldots,r-1$, solve for $c_r = d-c_1-\cdots-c_{r-1}$ and subtract the terms $c_r b^r = (c_rm,0,\ldots)$ from the polytopal constraints to obtain $\hat{\la}= \la - c_rb_r$ and removing the last vector $b^r$ altogether, leaving us with
\begin{align}\label{eq:pleth_reduced}
a^\la_{d,m} = \sum_{\sigma \in S_{K+1} } \sgn(\sigma) \sum_{r=1}^{4K^3+1} \sum_{ (c_1,\ldots,c_{r-1}) \in [1,2K]^{r-1}} \sum_{\bar{j} \in[K+1]^{r-2}} |Q(\bar{j}, c, \hat{\la} +\delta(K) -\sigma(\delta))|
\end{align}

By an analogous proof to the one of Proposition~\ref{prop:pleth_gen} we have the desired runtime bound.
\end{proof}

\begin{proof}[Proof of Theorem~\ref{thm:pleth}]
First, suppose that $k$ and $d$ are fixed. Then $k!$ and $(2k)^{d-1}$ and constants and  Proposition~\ref{prop:pleth_gen} gives that $a^\la_{d,m}$ can be computed in time $O((dm)^{dk}\log(dk)) = O(n^{dk}\log(d))$, which is of polynomial time. 

For the second part, Proposition~\ref{prop:fixed_aft2} gives us that ${\rm aft}(\la) \leq 4k^2$. Setting $K=4k^2$ we can then invoke Proposition~\ref{prop:pleth_aft}. 
\end{proof}

\section{Additional remarks}

\subsection{} We include the following, slightly edited, remark from Vojtech  Havl{\'\i}{\v{c}}ek (personal communication) on the implications of this work to quantum computing:

``The results in this paper have the following implications for the quantum algorithms proposed by Larocca and Havlicek in~\cite{LH24}. A superpolynomial quantum advantage means that a quantum algorithm runs in polynomial time while the best classical algorithm for the problem is expected to have a runtime that scales superpolynomially with the input size.

\begin{enumerate}
\item Unless Conjecture~\ref{conj:lr} (matching Hypothesis 1 in~\cite{LH24}) or the stronger Conjecture~\ref{conj:lr2} about the existence of a classical (possibly randomized) algorithm for computing Littlewood--Richardson coefficients is proved, there seems to be a narrow input regime in which their proposed quantum algorithm provides a super-polynomial quantum speedup. The set of inputs for which this is possible is however severely limited by Proposition~\ref{prop:LR_compute} and other regimes considered in Sections~\ref{sec:dimension},~\ref{sec:kostka}.

\item Theorem~\ref{thm:kron1} refutes Conjecture 2 in \cite{LH24} about the possibility of superpolynomial quantum speedups when computing Kronecker coefficients. It does not quite show that the quantum algorithm proposed in \cite{LH24}(and \cite{bravyi2024quantum}) has a classical polynomial time algorithm as that algorithm runs in time $O(f^\mu f^\nu/ f^\lambda \poly(n))$ and it is possible that it runs in polynomial time whenever $f^\lambda$, $f^\mu$ and $f^\nu$ scale super-polynomially and yet their ratio remains polynomial. However, there does not seem to be a nicely parametrized set of such partitions. Here it is posed  as Question~\ref{q:triples} and conjecture that there is classical algorithm with runtime $O(f^\mu f^\nu/ f^\lambda \poly(n))$. 

\item Theorem~\ref{thm:pleth} limits the set of inputs for which the plethysm coefficient quantum algorithm in \cite{LH24} gives a superpolynomial speedup. If Question~\ref{q:pleth} is resolved in affirmative, the algorithm from \cite{LH24} does not provide a superpolynomial quantum advantage. ''

\end{enumerate}

\subsection{} The careful reader would notice that we have provided a large class of triples $(\la,\mu,\nu)$ for which $g(\la,\mu,\nu)$ can be computed in polynomial time, so for these families of inputs we have $\textsc{ComputeKron} \in \FP \subset \SP$. In these cases the Kronecker coefficient is equal to the number of some objects, each of which can be computed (and thus verified) in polynomial time. This implies that for those cases the Kronecker coefficients have an efficient combinatorial interpretation. This interpretation can be derived from the algorithm, but it would not be very insightful and would not classify as nice. In fact, by~\cite{CDW,PPcomp} we already knew that $\textsc{ComputeKron} \in \FP \subset \SP$ whenever the number of rows $\ell(\la),\ell(\mu),\ell(\nu)$ are constant. Yet research into combinatorial interpretations for 2-, 3-,4- row partitions has continued past the earlier classical results in~\cite{rosas2001kronecker}, see~\cite{BMS,mishna2021vector,mishna2022estimating}. When one partition has two rows, but the others are not bounded, in certain cases we have criteria, see~\cite{ballantine2005kronecker} and~\cite{PPu}. 

\subsection{} There are even less available formulas or other results for the plethysm coefficients, see~\cite{colmenarejo2022mystery} for an overview and~\cite{orellana2022plethysm} for another representation theoretic interpretation. Special cases for  $a^\la_{d,m}$ are being considered when $d=2,3,4$ as in~\cite{orellana2024quasi}. Earlier work on the $a^\la_{d,m}$ in relation to Geometric Complexity theory was done in~\cite{IP17} and special cases of three-row plethysm coefficients were derived in~\cite{DIP}. In~\cite{DIP} we first used the approach described in Section~\ref{sec:pleth}, which was later used to show that computing plethysm coefficients is in $\GapP$ in~\cite{FI20}. The cases subject to Theorem~\ref{thm:pleth} give families when computing the plethysm coefficient is in $\FP \subset \SP$, and so should have a positive combinatorial interpretation. 

\subsection{} There is a curious parallel between Kronecker and plethysm coefficients. Even though they do not live in the same ``space'' informally speaking, they exhibit similar behavior and computational hardness. Using a quite indirect approach via Geometric Complexity Theory, Ikenmeyer and the author observed in~\cite{IP17} that in a certain stable limits we have $g(\la, (m^d), (m^d)) \geq a^\la_{d,m}$. It is also easy to see that $g(\la, (m^d),(m^d) ) = a^\la_{d,m}$ when $\la$ is a two-row partition. These values are actually the difference between successive numbers of partitions inside a rectangle, whose combinatorial proof of positivity by Kathy O'Hara~\cite{o1990unimodality} leads to a combinatorial interpretation (as observed in~\cite{PPu} and explicitly stated in~\cite{panova2023computational}). In the opposite direction, the multiplicity approach could lead to the desired symmetric chain decomposition as done in~\cite{orellana2024quasi} for $d\leq 4$. 

\subsection{} The LR and Kronecker coefficients are polynomially bounded in size whenever the three partitions involved have fixed Durfee square size (diagonal length) as shown in~\cite{PP22}. Trying to modify these proofs to compute the exact value runs into problems, in particular some exponentially large alternating sums. Even in the simplest case when $\nu =(n-a,1^a)$ is a hook, $d(\la),d(\mu) \leq k$, it is not clear how to efficiently compute $g(\la,\mu,\nu)$. The combinatorial interpretation of Blasiak-Liu \cite{Bla,BL,Liu} requires constructing exponentially many tableaux. We expect that there will still be a poly-time algorithm in this case.

\subsection{} While in this paper we were concerned with unary input (i.e. input size is equal to $n =|\lambda|$), in some cases binary input is also relevant. Binary input for partitions means that  we write the part sizes $\la_1,\la_2,\cdots$ in binary, so that the input size becomes $\log(\la_1) + \log(\la_2)+\cdots = O(\ell(\la) \log(\la_1) )$. With binary input and $\ell(\la),\ell(\mu),\ell(\nu)$ -- constants, we still have that $g(\la,\mu,\nu)$ can be computed in polynomial time, see~\cite{CDW,PPcomp}. It is thus natural to ask\footnote{as conjectured by M. Christandl, M. Walter, Personal Communication, 2023} whether the Kronecker coefficients's positivity is in $\QMA$ when the input is in binary (and there is no further restriction on lengths). 

\subsection{}\label{sec:rem_kron}
For most partitions the ratio $\frac{f^\la f^\mu}{f^\nu} = O(\sqrt{n!})$, since most partitions are close to the Plancherel shape. It is not hard to see that the Kronecker coefficients can be computed in time $\exp{O(n)}$, for example by using the character formula
$$g(\la,\mu,\nu) = \sum_{\alpha \vdash n} \frac{1}{z_\alpha} \chi^\la(\al) \chi^\mu(\al) \chi^\nu(\al),$$
where $\chi^\la(\al)$ are the irreducible characters of $S_n$ evaluated at a permutation of cycle type $\al$ and $z_\al$ is the size of the centralizer.
The character tables themselves can be computed via branching rules in time $\exp O(\sqrt{n})$, this approach is outlined in~\cite{pak2024signed}. 

The runtime bound also holds when the lengths of $\la,\mu,\nu$ are fixed since then the Kronecker coefficients are computable in $\poly(n)$ time. Thus the important cases are when $f^\la f^\mu /f^\nu = O(\poly(n))$, but the partitions themselves have large lengths. As we  saw in the examples in Section~\ref{sec:dimension} there are some interesting such cases. As a particular benchmark example we challenge the reader with considering $\la =(n-k,k)$ a two-row partition and $\mu$ and $\nu$ close to rectangular shapes in  a regime when $\frac{f^\mu f^\la}{f^\nu} = \Theta (\poly(n))$, see Example~\ref{ex:rect}.

\subsection{}\label{sec:rem_pleth}
Similar to the discussion above, plethysm coefficients can be computed in exponential time (using characters) and the challenge lies in finding polynomial time algorithms. 
While we were concerned only with $a^\la_{d,m}$, the same approach can be used to compute $a^\la_{\mu,\nu}$ in polynomial time when:
$|\mu|$ is fixed and $\ell(\la)$ is fixed (see~\cite{kahle2016plethysm}), or ${\rm aft}(\la)$ is fixed and $\mu$ and $\nu$ are arbitrary. In the second case we would expand $s_\nu$ via monomial quasisymmetric functions. The crux is to realize that only a small number of instances would contribute to $s_\la$ when ${\rm aft}(\la)$ is small. 

Contrary to that, if the inner partition size $m=|\nu|$ is fixed, even for $m=3$, the problem becomes $\SP$-hard~\cite{FI20}, i.e. no polynomial time algorithm would exist in general (assuming $\P \neq \NP$).

\subsection{}
It is possible to obtain all the coefficients in the expansion in terms of Schur functions (or other bases) ``at once'' as described in~\cite{barvinok1997sparse}. In our context it would mean, for example, obtaining all $g(\la,\mu,\nu)$ for $\la$ fixed and $\mu,\nu$ varying. However, the runtime of such an algorithm would be polynomial in the entire data (including the sizes of the coefficients) and for most cases this would mean superexponential in $n$.



\begin{thebibliography}{BCG{\etalchar{+}}24}

\bibitem[Aar16]{Aa}
Scott Aaronson.
\newblock {P}$=^?${NP}.
\newblock {\em Open problems in mathematics}, pages 1--122, 2016.


\bibitem[BO05]{ballantine2005kronecker}
Cristina~M Ballantine and Rosa~C Orellana.
\newblock On the {K}ronecker product $s_{(n-p, p)}\ast s_\lambda$.
\newblock {\em the electronic journal of combinatorics}, pages R28--R28, 2005.

\bibitem[Bar94]{barvinok1994polynomial}
Alexander~I Barvinok.
\newblock A polynomial time algorithm for counting integral points in polyhedra
  when the dimension is fixed.
\newblock {\em Mathematics of Operations Research}, 19(4):769--779, 1994.

\bibitem[BF97]{barvinok1997sparse}
Alexander Barvinok and Sergey Fomin.
\newblock Sparse interpolation of symmetric polynomials.
\newblock {\em Advances in Applied Mathematics}, 18(3):271--285, 1997.


\bibitem[Bla17]{Bla}
Jonah Blasiak.
\newblock {K}ronecker coefficients for one hook shape.
\newblock {\em S{\'e}m. Lothar. Combin}, 77:2016--2017, 2017.

\bibitem[BL18]{BL}
Jonah Blasiak and Ricky~Ini Liu.
\newblock {K}ronecker coefficients and noncommutative super {S}chur functions.
\newblock {\em Journal of Combinatorial Theory, Series A}, 158:315--361, 2018.

\bibitem[BMS15]{BMS}
Jonah Blasiak, Ketan Mulmuley, and Milind Sohoni.
\newblock {G}eometric {C}omplexity {T}heory {I}{V}: nonstandard quantum group
  for the {K}ronecker problem.
\newblock In {\em Memoirs of the American Mathematical Society}, volume 235.
  2015.

\bibitem[BCG{\etalchar{+}}24]{bravyi2024quantum}
Sergey Bravyi, Anirban Chowdhury, David Gosset, Vojt{\v{e}}ch
  Havl{\'\i}{\v{c}}ek, and Guanyu Zhu.
\newblock Quantum complexity of the {K}ronecker coefficients.
\newblock {\em PRX Quantum}, 5(1):010329, 2024.







\bibitem[Bri93]{brion1993stable}
Michel Brion.
\newblock Stable properties of plethysm: on two conjectures of foulkes.
\newblock {\em Manuscripta mathematica}, 80:347--371, 1993.

\bibitem[CDW12]{CDW}
Matthias Christandl, Brent Doran, and Michael Walter.
\newblock Computing multiplicities of {L}ie group representations.
\newblock In {\em 2012 IEEE 53rd Annual Symposium on Foundations of Computer
  Science}, pages 639--648. IEEE, 2012.

\bibitem[CHW15]{CHW}
Matthias Christandl, Aram Harrow, and Michael Walter.
\newblock presentation: On the computational complexity of the membership
  problem for moment polytopes.
\newblock In {\em Algorithms and Complexity in Algebraic Geometry
  Reunion-Workshop}. Simons Institute, Berkeley, 2015.

\bibitem[Col17]{colmenarejo2017stability}
Laura Colmenarejo.
\newblock Stability properties of the plethysm: a combinatorial approach.
\newblock {\em Discrete Mathematics}, 340(8):2020--2032, 2017.

\bibitem[COS{\etalchar{+}}24]{colmenarejo2022mystery}
Laura Colmenarejo, Rosa Orellana, Franco Saliola, Anne Schilling, and Mike
  Zabrocki.
\newblock The mystery of plethysm coefficients.
\newblock In {\em Open Problems in Algebraic Combinatorics}, volume 110 of {\em
  Proc. Sympos. Pure Math.}, pages 275--292. Amer. Math. Soc., Providence, RI,
  2024.

\bibitem[DIP20]{DIP}
Julian D\"orfler, Christian Ikenmeyer, and Greta Panova.
\newblock On geometric complexity theory: Multiplicity obstructions are
  stronger than occurrence obstructions.
\newblock {\em SIAM Journal on Applied Algebra and Geometry}, 4(2):354--376,
  2020.

\bibitem[FI20]{FI20}
Nick Fischer and Christian Ikenmeyer.
\newblock The computational complexity of plethysm coefficients.
\newblock {\em computational complexity}, 29:1--43, 2020.

\bibitem[Ful97]{Fu97}
  William Fulton.
  \newblock {\em Young tableaux: with applications to representation theory and geometry}, number 35.
  \newblock Cambridge University Press, 1997


\bibitem[GM16]{GM}
Antonio Giambruno and Sergey Mishchenko.
\newblock Degrees of irreducible characters of the symmetric group and
  exponential growth.
\newblock {\em Proceedings of the American mathematical society},
  144(3):943--953, 2016.
  
  \bibitem[Ike16]{ikenmeyer2016small}
Christian Ikenmeyer.
\newblock Small Littlewood--Richardson coefficients.
\newblock {\em Journal of Algebraic Combinatorics}, 44:1--29, 2016.

\bibitem[IMW17]{IMW17}
Christian Ikenmeyer, Ketan Mulmuley, and Michael Walter.
\newblock On vanishing of {K}ronecker coefficients.
\newblock {\em computational complexity}, 26:949--992, 2017.


\bibitem[IP17]{IP17}
Christian Ikenmeyer and Greta Panova.
\newblock Rectangular {K}ronecker coefficients and plethysms in geometric
  complexity theory.
\newblock {\em Advances in Mathematics}, 319:40--66, 2017.

\bibitem[IP24]{IP24}
Christian Ikenmeyer and Greta Panova.
\newblock All {K}ronecker coefficients are reduced {K}ronecker coefficients.
\newblock In {\em Forum of Mathematics, Pi}, volume~12, page e22. Cambridge
  University Press, 2024.

\bibitem[IS23]{ikenmeyer2023remark}
Christian Ikenmeyer and Sathyawageeswar Subramanian.
\newblock A remark on the quantum complexity of the {K}ronecker coefficients.
\newblock {\em arXiv:2307.02389}, 2023.


\bibitem[IP01]{impagliazzo2001complexity}
Russell Impagliazzo and Ramamohan Paturi.
\newblock On the complexity of k-sat.
\newblock {\em Journal of Computer and System Sciences}, 62(2):367--375, 2001.

\bibitem[KM16]{kahle2016plethysm}
Thomas Kahle and Mateusz Micha{\l}ek.
\newblock Plethysm and lattice point counting.
\newblock {\em Foundations of Computational Mathematics}, 16(5):1241--1261,
  2016.

\bibitem[LH24]{LH24}
Martin Larocca and Vojtech Havlicek.
\newblock Quantum algorithms for representation-theoretic multiplicities.
\newblock {\em arXiv:2407.17649}, 2024.

\bibitem[Liu17]{Liu}
Ricky Liu.
\newblock A simplified {K}ronecker rule for one hook shape.
\newblock {\em Proceedings of the American Mathematical Society},
  145(9):3657--3664, 2017.

\bibitem[Mac98]{Mac}
Ian~G. Macdonald.
\newblock {\em Symmetric functions and Hall polynomials}.
\newblock Oxford university press, 2. edition, 1998.


\bibitem[MRS21]{mishna2021vector}
Marni Mishna, Mercedes Rosas, and Sheila Sundaram.
\newblock Vector partition functions and {K}ronecker coefficients.
\newblock {\em Journal of Physics A: Mathematical and Theoretical},
  54(20):205204, 2021.

\bibitem[MT22]{mishna2022estimating}
Marni Mishna and Stefan Trandafir.
\newblock Estimating and computing {K}ronecker coefficients: a vector partition
  function approach.
\newblock {\em arXiv:2210.12128}, 2022.


\bibitem[MPP18]{MPPa}
Alejandro~H Morales, Igor Pak, and Greta Panova.
\newblock Asymptotics of the number of standard young tableaux of skew shape.
\newblock {\em European Journal of Combinatorics}, 70:26--49, 2018.

\bibitem[O'H90]{o1990unimodality}
Kathleen~M O'Hara.
\newblock Unimodality of gaussian coefficients: a constructive proof.
\newblock {\em Journal of Combinatorial Theory, Series A}, 53(1):29--52, 1990.

\bibitem[OSSZ22]{orellana2022plethysm}
Rosa Orellana, Franco Saliola, Anne Schilling, and Mike Zabrocki.
\newblock Plethysm and the algebra of uniform block permutations.
\newblock {\em Algebraic Combinatorics}, 5(5):1165--1203, 2022.

\bibitem[OSSZ24]{orellana2024quasi}
Rosa Orellana, Franco Saliola, Anne Schilling, and Mike Zabrocki.
\newblock From quasi-symmetric to schur expansions with applications to
  symmetric chain decompositions and plethysm.
\newblock {\em Electron. J. Combin.}, 31(4), 2024.

\bibitem[Pak24]{Pak22}
Igor Pak.
\newblock What is a combinatorial interpretation?
\newblock In {\em Open Problems in Algebraic Combinatorics}, volume 110 of {\em
  Proc. Sympos. Pure Math.}, pages 191--260. Amer. Math. Soc., Providence, RI,
  2024.



\bibitem[PP14]{PPu}
Igor Pak and Greta Panova.
\newblock Unimodality via {K}ronecker products.
\newblock {\em Journal of Algebraic Combinatorics}, 40:1103--1120, 2014.

\bibitem[PP17]{PPcomp}
Igor Pak and Greta Panova.
\newblock On the complexity of computing {K}ronecker coefficients.
\newblock {\em Comput.\ Complexity}, 26:1--36, 2017.

\bibitem[PP20]{PP20}
 Igor Pak and Greta Panova. 
 \newblock Bounds on Kronecker coefficients via contingency tables.
 \newblock {\em Linear Algebra and its Applications}, 602: 157--178, 2020.


\bibitem[PP23]{PP22}
Igor Pak and Greta Panova.
\newblock Durfee squares, symmetric partitions and bounds on {K}ronecker
  coefficients.
\newblock {\em Journal of Algebra}, 629:358--380, 2023.

\bibitem[PPY19]{PPY}
Igor Pak, Greta Panova, and Damir Yeliussizov.
\newblock On the largest {K}ronecker and {L}ittlewood--{R}ichardson
  coefficients.
\newblock {\em Journal of Combinatorial Theory, Series A}, 165:44--77, 2019.

\bibitem[PR24]{pak2024signed}
Igor Pak and Colleen Robichaux.
\newblock Signed combinatorial interpretations in algebraic combinatorics.
\newblock {\em arXiv:2406.13902}, 2024.

\bibitem[PV05]{PV05}
 Igor Pak and Ernesto Vallejo. 
 \newblock Combinatorics and geometry of Littlewood–Richardson cones.
 \newblock {\em European Journal of Combinatorics}, 26.6: 995--1008, 2005.

 \bibitem[Pan23]{panova2023computational}
Greta Panova.
\newblock Computational complexity in algebraic combinatorics.
\newblock In {\em Current Developments in Mathematics}. Harvard University
  Press, 2023.

\bibitem[Pan24]{Pan23}
Greta Panova.
\newblock Complexity and asymptotics of structure constants.
\newblock In {\em Open Problems in Algebraic Combinatorics}, volume 110 of {\em
  Proc. Sympos. Pure Math.}, pages 61--86. Amer. Math. Soc., Providence, RI,
  2024.

\bibitem[Reg98]{regev}
Amitai Regev.
\newblock Maximal degrees for young diagrams in the (k, l) hook.
\newblock {\em European Journal of Combinatorics}, 19(6):721--726, 1998.

\bibitem[Ros01]{rosas2001kronecker}
Mercedes~H Rosas.
\newblock The {K}ronecker product of {S}chur functions indexed by two-row
  shapes or hook shapes.
\newblock {\em Journal of algebraic combinatorics}, 14:153--173, 2001.

\bibitem[Sag13]{Sag}
Bruce~E. Sagan.
\newblock {\em The symmetric group: representations, combinatorial algorithms,
  and symmetric functions}, volume 203.
\newblock Springer, 2013.

\bibitem[Sta97]{S1}
Richard Stanley.
\newblock {\em Enumerative Combinatorics}, volume~1 and volume~2.
\newblock Cambridge University Press, 2 edition, 1997.

\bibitem[Sta00]{Sta00}
Richard Stanley.
\newblock Positivity {P}roblems and {C}onjectures in {A}lgebraic
  {C}ombinatorics.
\newblock {\em Mathematics: Frontiers and Perspectives}, pages 295--319, 2000.

\bibitem[Sta12]{StaMO}
Richard Stanley.
\newblock Skew {K}ostka coefficients from {L}ittlewood--{R}ichardson Coefficients.
\newblock {\url{https://mathoverflow.net/questions/116171/skew-kostka-coefficients-from-littlewood-richardson-coefficients}}, 2012.

\bibitem[Wig19]{Wig}
Avi Wigderson.
\newblock {\em Mathematics and computation: A theory revolutionizing technology
  and science}.
\newblock Princeton University Press, 2019.

\end{thebibliography}

\newcommand{\etalchar}[1]{$^{#1}$}

\end{document}